\documentclass[11pt]{article}

\usepackage{amsfonts}
\usepackage{amssymb}
\usepackage{amstext}
\usepackage{amsmath}
\usepackage{enumerate}
\usepackage{algorithm}
\usepackage{algorithmic}
\usepackage[normalem]{ulem}
\usepackage{cancel}

\usepackage{multirow, makecell}
\usepackage{tablefootnote}

\usepackage{amsthm}

\usepackage{thmtools}
\usepackage{thm-restate}

\usepackage[pagebackref,colorlinks=true,pdfpagemode=UseNone,citecolor=OliveGreen,linkcolor=BrickRed,urlcolor=BrickRed,pdfstartview=FitH]{hyperref}
\usepackage{cleveref}

\usepackage[usenames,dvipsnames,svgnames,table]{xcolor}

\usepackage{graphicx}
\usepackage{subcaption}

\usepackage{framed}
\usepackage{enumitem}


\DeclareMathAlphabet{\mathsfit}{T1}{\sfdefault}{\mddefault}{\sldefault}
\SetMathAlphabet{\mathsfit}{bold}{T1}{\sfdefault}{\bfdefault}{\sldefault}

\usepackage{circledsteps}

\newtheorem{theorem}{Theorem}[section]

\newtheorem{lemma}[theorem]{Lemma}

\newtheorem{definition}[theorem]{Definition}

\newtheorem{proposition}[theorem]{Proposition}

\newtheorem*{claim*}{Claim}

\newtheorem*{problem*}{Problem}

\newtheorem{remark}[theorem]{Remark}
\newtheorem*{remark*}{Remark}

\newtheorem{observation}[theorem]{Observation}
\newtheorem*{observation*}{Observation}

\numberwithin{equation}{section}
\numberwithin{table}{section}

\newcommand{\R}{\ensuremath{\mathbb R}}
\newcommand{\mZ}{\ensuremath{\mathbb Z}}

\newcommand{\E}[1]{{\mathbb{E}}\left[#1\right]}

\newcommand{\junk}[1]{}

\newcommand{\norm}[1]{\left\lVert#1\right\rVert}

\def\eps{{\varepsilon}}

\DeclareMathAlphabet{\mathsfit}{T1}{\sfdefault}{\mddefault}{\sldefault}
\SetMathAlphabet{\mathsfit}{bold}{T1}{\sfdefault}{\bfdefault}{\sldefault}

\def\vc{{\mathsfit c}}

\def\vu{{\mathsfit u}}
\def\vv{{\mathsfit v}}
\def\vw{{\mathsfit w}}
\def\vx{{\mathsfit x}}

\def\vz{{\mathsfit z}}

\def\mA{{\mathsfit A}}
\def\mB{{\mathsfit B}}

\def\mI{{\mathsfit I}}

\def\mX{{\mathsfit X}}
\def\mY{{\mathsfit Y}}
\def\mZ{{\mathsfit Z}}

\def\tr{\operatorname{tr}}

\global\long\def\E{\mathbb{E}}
\global\long\def\R{\mathbb{R}}

\newcommand{\inner}[2]{\langle #1, #2 \rangle} 

\renewcommand{\succeq}{\succcurlyeq}

\title{Experimental Design for Any $p$-Norm}

\author{
Lap Chi Lau\footnote{School of Computer Science, University of Waterloo. Email: \href{mailto:lapchi@uwaterloo.ca}{lapchi@uwaterloo.ca}},~~~~~
Robert Wang\footnote{School of Computer Science, University of Waterloo.
Email: \href{mailto:robert.wang2@uwaterloo.ca}{robert.wang2@uwaterloo.ca}},~~~~~
Hong Zhou\footnote{School of Mathematics and Statistics, Fuzhou University. Email: \href{mailto:hong.zhou@fzu.edu.cn}{hong.zhou@fzu.edu.cn}}
}

\date{}

\begin{document}

\maketitle

\begin{abstract}
We consider a general $p$-norm objective for experimental design problems that captures some well-studied objectives (D/A/E-design) as special cases.    
We prove that a randomized local search approach provides a unified algorithm to solve this problem for all $p$.
This provides the first approximation algorithm for the general $p$-norm objective, and a nice interpolation of the best known bounds of the special cases.
\end{abstract}

\newpage

\section{Introduction}

In experimental design problems, we are given vectors $\vv_1, \ldots, \vv_n \in \R^d$ and a budget $k \geq d$, and the goal is choose a subset $S$ of $k$ vectors so that $\sum_{i \in S} \vv_i \vv_i^\top$ optimizes some objective function that measures the ``diversity'' of the input data.
The most popular and well-studied objective functions are:
\begin{itemize}
    \item D-design: Maximizing $(\det(\sum_{i \in S} \vv_i \vv_i^\top))^{\frac1d}$. 
    \item A-design: Minimizing $\frac1d \tr((\sum_{i \in S} \vv_i \vv_i)^{-1})$. 
    \item E-design: Maximizing $\lambda_{\min}(\sum_{i \in S} \vv_i \vv_i^\top)$. 
\end{itemize}
Experimental design problems have a long history and a wide range of applications, from statistics to machine learning to numerical linear algebra to graph algorithms.
For more information about these applications, we refer the reader to~\cite{SX18,NST19,AZLSW21,LZ20,LZ21} and the references therein.

Although the objectives of D/A/E-design look quite different, 
we observe that there is a natural generalization using eigenvalues that captures all three objective functions as special cases.

\begin{definition}[$p$-Norm of Inverse Eigenvalues]
Given a $d$-dimensional real-symmetric matrix $\mA$ with eigenvalues $\lambda_1, \ldots, \lambda_d > 0$ and a natural number $0 \leq p \leq \infty$, we define
\begin{equation} \label{eq:Phi-p}
    \Phi_p(\mA) := \bigg( \frac{1}{d} \tr\big( \mA^{-p}\big) \bigg)^{\frac1p} = \bigg( \frac{1}{d} \sum_{i=1}^d \lambda_i^{-p} \bigg)^{\frac1p},
\end{equation}
with $\Phi_0(\mA) := \lim_{p \to 0^+} \Phi_p(\mA)$ and $\Phi_{\infty}(\mA) := \lim_{p \to +\infty} \Phi_p(\mA)$.
\end{definition}

Given $p \geq 0$, we refer to the experimental design problem with objective function $\Phi_p$ as $\Phi_p$-design.
To see that $\Phi_p$-design is a generalization of D/A/E-design, let $\mA = \sum_{i \in S} \vv_i \vv_i^T$ and note that:
\begin{itemize}
    \item For $p = \infty$, $\Phi_\infty(\mA) = \lambda_{\max}(\mA^{-1}) = 1/\lambda_{\min}(\mA)$, which is the inverse of the E-design objective;
    \item For $p=1$, $\Phi_1(\mA) = \frac1d \tr(\mA^{-1})$ is exactly the A-design objective;
    \item For $p=0$, $\Phi_0(\mA)$ is the inverse of the D-design objective, as
    \[
    \Phi_0(\mA) = \lim_{p \to 0^+} \bigg(\frac1d \sum_{i=1}^d \lambda_i^{-p} \bigg)^{1/p} = \bigg( \prod_{i=1}^d \lambda_i^{-1} \bigg)^{1/d} = \det(\mA)^{-1/d},
    \]
    where the second equality is a well-known fact (see, e.g., Exercise 28 in Chapter 5 of~\cite{WZ15}).

\end{itemize}

It is known that $\Phi_p(\mA)$ is convex in $\mA$ for any given $0 \leq p \leq \infty$, and so the following is the natural convex programming relaxation for $\Phi_p$-design:
\begin{equation} \label{eq:convex-p}
    \begin{aligned}
        & \underset{\vx \in \R^n}{\rm minimize} & & \Phi_p\bigg( \sum_{i=1}^n \vx(i) \cdot \vv_i \vv_i^\top \bigg) \\ 
        & \text{\rm subject to} & & \sum_{i=1}^n \vx(i) \leq k, \\
        & & & 0 \leq \vx(i) \leq 1, \quad ~\text{ for } 1 \leq i \leq n.
    \end{aligned} 
\end{equation}

To the best of our knowledge,
there are no known approximation algorithms for the general $\Phi_p$-design problem,
other than the special cases $p=0,1,\infty$ which we summarize as follows (the notation $x \gtrsim y$ denotes that $x \geq c y$ for some large enough constant $c$):
\begin{enumerate}
    \item There is a $(1+\eps)$-approximation algorithm for $\Phi_0$-design (D-design) when $k \gtrsim d/\eps$~\cite{SX18,MSTX19,NST19,LZ21}.
    \item There is a $(1+\eps)$-approximation algorithm for $\Phi_1$-design (A-design) when $k \gtrsim d/\eps$~\cite{MSTX19,NST19,LZ21}.
    \item There is a $(1+\eps)$-approximation algorithm for $\Phi_{\infty}$-design (E-design) when $k \gtrsim d/\eps^2$~\cite{AZLSW17c,AZLSW21,LZ20}.
\end{enumerate}
These results are tight in the sense that they match the known integrality gap lower bound of the convex programming relaxation \eqref{eq:convex-p} (see~\cite{NST19} for integrality gap examples).

Note that there is a $d/\eps$ vs $d/\eps^2$ gap between the relaxations for D/A-design ($p=0,1$) and for E-design ($p=\infty$). 
The main question that we study in this paper is:
How does the integrality gap of the convex programming relaxation \eqref{eq:convex-p} change with varying value of $p$? 
In particular, where does the transition from $d/\eps$ to $d/\eps^2$ happen?

\subsection{Main Result}

Our main result is that, when $k \gtrsim \min\{dp/\eps, d/\eps^2 \}$, there is a $(1+\eps)$-approximation algorithm for $\Phi_p$-design.

\begin{theorem} \label{t:main}
Given an integer $p \geq 1$, let $\vx \in [0,1]^n$ be an optimal fractional solution to \eqref{eq:convex-p}. 
For any $\eps \in (0,1)$, let $\gamma = \max\{\eps, 1/p\}$, if $k \gtrsim d / (\gamma \eps)$,
then there is a randomized polynomial time algorithm which returns an integral solution $\mZ = \sum_{i=1}^n \vz(i) \cdot \vv_i \vv_i^\top$ with 
$\vz(i) \in \{0,1\}$ for $1 \leq i \leq n$ such that
\[
\Phi_p\bigg( \sum_{i=1}^n \vz(i) \cdot \vv_i \vv_i^\top \bigg) \leq (1+\eps) \cdot \Phi_p\bigg( \sum_{i=1}^n \vx(i) \cdot \vv_i \vv_i^\top\bigg)
\quad \textrm{and} \quad
\sum_{i=1}^n \vz(i) \leq k,
\]
with probability at least $1 - O\left(\frac{(d^2+d/\gamma)^2}{\eps^2 p^2} \cdot e^{-\Omega( \gamma \sqrt{d})}\right) - e^{-\Omega(\eps d/\gamma)}$.
\end{theorem}

\begin{remark*}
The result in Theorem~\ref{t:main} can be generalized to all real $p \geq 1$, see Remark~\ref{r:real_p} in Section~\ref{s:objective} for more details.
\end{remark*}

This is the first approximation algorithm for $\Phi_p$-design for general $p$.
\Cref{t:main} shows that $\Phi_p$-design for constant $p$ admits as good an approximation algorithm as for D/A-design, and there is a unifying algorithm to achieve this guarantee.

Note that, when $p \to +\infty$, $\Phi_p$ becomes the E-design objective and $\gamma = \max\{\eps, 1/p\} = \eps$. Thus, \Cref{t:main} provides a nice interpolation between the $d/\eps$ bound for D/A-design and the $d/\eps^2$ bound for E-design.

We further remark that our results can be generalized to the weighted setting to handle multiple budget/knapsack constraints as in~\cite{LZ21}, but we omit the details to keep the presentation cleaner as they are the same as in~\cite{LZ21}.

The proof of \Cref{t:main} is built on the randomized local search approach in~\cite{LZ21} and~\cite{LZ20}, but several new technical ideas are needed to handle higher moments that are introduced by higher $p$-norm.
In \Cref{s:framework}, we will review the background and previous work, present the algorithm and the overall structure of the analysis, and explain the new ideas in this work.
Then, in \Cref{s:objective}, we will present the details for $p$-norm experimental design.

\subsection{Related Work}

In this paper, we are focusing on generalizing D/A/E-design with the $p$-norm objective in \eqref{eq:Phi-p}. 
However, there are other different ways to extend D/A/E-design. 
For example, the Bayesian framework of experimental design~\cite{CV95} extends the problem by adding a fixed matrix $\mB$ (which encodes some prior information, e.g., a multiple of identity) to the covariance matrix before applying the objective function. When $\mB = 0$, we recover classical experimental design problems.
Tantipongpipat~\cite{Tan20} provided an approximation algorithm for Bayesian A-design problem when $\mB$ is a multiple of identity. 
Dereziński et. al.~\cite{DLM20} studied approximation algorithms for Bayesian A/D-design (together with C/V-design, i.e., two other objectives) with a general PSD matrix $\mB$.
Another example is using elementary symmetric polynomials to generalize A/D-design~\cite{MS17}. Nikolov et. al.~\cite{NST19} extended their approximation algorithm for A-design to tackle this family of generalized ratio objectives.

All these are interesting generalizations of D/A/E-design.
It would be interesting to see whether the randomized local search approach in~\cite{LZ20,LZ21} and this paper can be applied to these settings to get better results.

\section{The Framework} \label{s:framework}

In this section, we first review the randomized local search approach in \Cref{ss:previous}.
Then, we present the full algorithm in \Cref{ss:algorithm} and state the main technical theorem.
Then, in \Cref{ss:plan}, we provide the overall proof plan and the precise statements for analyzing the randomized exchange algorithm, and then put together the statements to prove the main technical theorem.
Finally, we discuss the main difficulty and the new ideas needed to analyze the $\Phi_p$ objective in \Cref{s:objective}.

\subsection{Randomized Local Search Approach} \label{ss:previous}

The proof of \Cref{t:main} is built on the randomized local search approach in~\cite{LZ21},
which is based on the regret minimization framework developed in~\cite{AZLSW21} for experimental design problems and the randomized spectral rounding techniques in~\cite{LZ20}.
We review this approach in this subsection.

In~\cite{LZ21}, the first step is to solve \eqref{eq:convex-p} for D/A/E-design to obtain a solution $\vx \in \R^n$, and then to normalize the vectors $\vv_i$'s so that $\sum_{i=1}^n \vx(i) \cdot \vv_i \vv_i^\top = \mI$ for using the regret minimization framework in~\cite{AZLSW21}.
Then, the rounding algorithm starts from a random initial solution set $S_0$ that is independently sampled according to $\vx$.
Using the density matrix $\mA_t$ maintained by the regret minimization framework at each step $t$, the algorithm randomly chooses a pair of vectors $\vv_{i_t}$ and $\vv_{j_t}$ with the following probability distributions:
\begin{equation} \label{e:E-distribution}
\begin{aligned}
& \Pr(i_t = i) \propto \big(1-\vx(i)\big) \cdot \Big(1- \alpha \inner{\vv_i \vv_i^\top}{\mA_t^{1/2}}\Big)
{\rm~~and~~}
\Pr(j_t = j) \propto \vx(j) \cdot \Big(1+ \alpha \inner{\vv_j \vv_j^\top}{\mA_t^{1/2}}\Big),
\end{aligned}
\end{equation}
and set $S_{t+1} \gets S_t -\vv_{i_t} + \vv_{j_t}$.
Using the above randomized local search strategy, it can be shown that the size of the solution is expected to stay around $k$, while the potential function from the regret minimization framework~\cite{AZLSW21} related to the minimum eigenvalue is expected to improve as long as the minimum eigenvalue is less than $1-\eps$.
Freedman's martingale inequality and a new concentration inequality for non-martingales are used in~\cite{LZ20} to prove that all these quantities are close to their expected values with high probability.
This randomized local search algorithm provides a $(1 \pm \eps)$-approximate solution for D/A/E-design when $k \gtrsim d/\eps^2$.

The main contribution in~\cite{LZ21} is to prove that the randomized local search algorithm can be adapted to achieve $(1 \pm \eps)$-approximation for D/A-design when $k \gtrsim d /\eps$, thus providing a unifying approach to achieve the optimal integrality gap for D/A/E-design.
Essentially, the algorithm is the same as the one for E-design but only require that the solution to have minimum eigenvalue $3/4$ rather than $1-\eps$.
For the analysis, the randomized local search algorithm is conceptually divided into two phases.
In the first phase, the algorithm will find a solution with minimum eigenvalue at least $3/4$ with high probability when $k \gtrsim d/\eps$.
In the second phase, the minimum eigenvalue will maintain to be at least $1/4$ with high probability, and the objective value for D/A-design will improve to $(1\pm\eps)$ times the optimal objective value in polynomial time with high probability.
The analysis of the first phase follows directly from earlier work on spectral rounding in~\cite{LZ20}.
The analysis of the second phase includes two main parts:
(1) to show that, in expectation, the probability distributions in~(\ref{e:E-distribution}) for E-design are also good for making progress towards D/A-design objectives;
(2) to show that the progress in the objective value is concentrated around the expectation with a martingale concentration argument.
The condition that the minimum eigenvalue is at least $1/4$ is very important in both parts in the second phase, and the optimality conditions for the convex program~\eqref{eq:convex-p} is crucially used in the second part in the second phase.

In this paper, we will extend the algorithm and the analysis described above, and show that the randomized local search algorithm provides a unifying approach for $\Phi_p$-design for all $p$.

\subsection{The Algorithm} \label{ss:algorithm}

We present the full algorithm for $\Phi_p$-design in this subsection.

\begin{framed}
\noindent {\bf Randomized Exchange Algorithm}

\noindent Input: $m$ vectors $\vu_1, ..., \vu_m \in \R^d$, a budget $k \geq d$, an accuracy parameter $\eps \in (0,1)$.

    \begin{enumerate}
    
    \item Solve the convex programming relaxation~\eqref{eq:convex-p} and obtain an optimal solution $\vx \in [0,1]^m$ with at most $d^2 + 1$ fractional entries, i.e.~$|\{i \in [m] \mid 0 < \vx(i) < 1\}| \leq d^2 + 1$.
    Let $\mX = \sum_{i=1}^m \vx(i) \cdot \vu_i \vu_i^\top$. 

    \item Preprocessing: 
	Let $\vv_i \leftarrow \mX^{-1/2} \vu_i$ for all $i \in [m]$, 
	so that $\sum_{i=1}^m \vx(i) \cdot \vv_i \vv_i^\top = \mI_d$.
    
    \item Initialization: $t \gets 1$, $S_0 \gets \emptyset$, $\gamma = \max\{\frac{\eps}{6}, \frac{1}{6p}\}$, $\kappa=\max\{ \frac{\eps}{2}, \frac{1}{2p}\}$, $M \gets \frac{d}{\gamma} + d^2 + 1$, and $\alpha \gets \frac{\sqrt{d}}{\gamma}$.
    
    \item Add $i$ into $S_0$ independently with probability $\vx(i)$ for each $i \in [m]$. Let $\mY_1 \gets \sum_{i \in S_0} \vu_i \vu_i^\top$ and $\mZ_1 \gets \sum_{i \in S_0} \vv_i \vv_i^\top$.
    
    \item While the termination condition
    $ (\tr( \mY_t^{-p} ) )^{\frac1p} \leq (1+\eps) (\tr(\mX^{-p}))^{\frac1p} $ 
    is not satisfied and $t \leq \frac{2M}{\gamma} + \frac{2M}{\eps p}$,
    do the following:
    \begin{enumerate}
        \item $S_t \gets$ Exchange($S_{t-1}$).
        \item Set $\mY_{t+1} \gets \sum_{i \in S_t} \vu_i \vu_i^\top$, $\mZ_{t+1} \gets \sum_{i \in S_t} \vv_i \vv_i^\top$ and $t \gets t+1$.
    \end{enumerate}
    
    \item Return $S_{t-1}$ as the solution. 
\end{enumerate}
\end{framed}

\vspace{1mm}

The main algorithm is almost the same as the one in~\cite{LZ21}, but with two additional parameters $\gamma, \kappa$ that will depend on the value of $p$. 
For A-design (when $p=1$), the algorithm in~\cite{LZ21} is just a special case of the randomized exchange algorithm with parameter $\gamma = \frac18$ and $\kappa = 1$.
In this paper, the parameter $\gamma$ is used to adjust the learning rate $\alpha$ of the regret minimization framework and the parameter $\kappa$ will be used in the exchange subroutine, which is described as follows.

\vspace{1mm}
\begin{framed}
\noindent {\bf Exchange Subroutine}

\noindent Input: the current solution set $S_{t-1}$.

    \begin{enumerate}
        \item Compute the action matrix $\mA_t := (\alpha \mZ_t - c_t \mI)^{-2}$, where $\mZ_t = \sum_{i \in S_{t-1}} \vv_i \vv_i^\top$ and $c_t$ is the unique scalar such that $\mA_t \succ 0$ and $\tr(\mA_t) = 1$.
        \item Define $S'_t := \big\{ i \in S_{t-1} \mid \alpha \inner{\vv_i \vv_i^\top}{\mA_t^{1/2}} \leq \frac12~\text{and}~\inner{\vv_i \vv_i^\top}{\mZ_t^{-1}} \leq \kappa \big\}$.
        
        \item Sample $j_t \in [m] \backslash S_{t-1}$ from the following probability distribution
        \begin{align*}
            & \Pr(j_t = j) =\frac{\vx(j)}{M} \cdot \big(1+ \alpha \inner{\vv_j \vv_j^\top}{\mA_t^{1/2}}\big), \text{ for } j  \in [m] \backslash  S_{t-1} \text{ and} \\ 
            & \Pr(j_t = \emptyset) = 1- \sum_{j \in [m] \backslash S_{t-1}} \frac{\vx(j)}{M} \cdot \big(1+ \alpha \inner{\vv_j \vv_j^\top}{\mA_t^{1/2}}\big).
        \end{align*}
        
        \item Sample $i_t \in S'_{t-1}$ from the following probability distribution 
        \begin{align*}
            & \Pr(i_t = i) =\frac{1-\vx(i)}{M} \cdot \big(1-\alpha \inner{\vv_i \vv_i^\top}{\mA_t^{1/2}}\big), \text{ for } i \in S'_{t-1} \text{ and} \\
            & \Pr(i_t = \emptyset) = 1- \sum_{i \in S'_{t-1}} \frac{1-\vx(i)}{M} \cdot \big(1-\alpha \inner{\vv_i \vv_i^\top}{\mA_t^{1/2}}\big).
        \end{align*}
        
        \item Return $S_t := S_{t-1} \cup \{ j_t\} \backslash \{i_t\}$.
    \end{enumerate}
\end{framed}

\vspace{1mm}

The exchange subroutine is also almost the same as in~\cite{LZ21}.
The key difference is to use the new parameter $\kappa$ to further restrict the set of vectors that are allowed to remove from the current solution.

\begin{remark} \label{r:exponential}
Using the new analysis in this paper, the same algorithm in~\cite{LZ21} (without changing the parameters $\gamma$ and $\kappa$) can achieve $(1+\eps)$-approximation for $\Phi_p$-design when $k \gtrsim 2^{O(p)} d/\eps$, with an exponential dependence on $p$ in the budget requirement, much worse than $k \gtrsim \min\{dp/\eps, d/\eps^2\}$ in \Cref{t:main}.
\end{remark}

\subsubsection{Main Technical Theorem}

To prove \Cref{t:main}, we will prove that the randomized exchange algorithm is a bicriteria approximation algorithm,
such that it returns a solution that is $(1+\eps)$-approximate in the $\Phi_p$ objective and the size of the solution is not much larger than $k$. 

\begin{theorem} \label{t:main-p-norm}
Given $\eps \in (0,1)$,  
if $k \gtrsim \frac{d}{\gamma  \eps}$, 
then the randomized exchange algorithm returns a solution set $S$ within $\frac{2M}{\gamma} + \frac{2M}{\eps p}$ iterations such that
\[
\Big( \tr\Big( \Big(\sum_{i \in S} \vu_i \vu_i^\top \Big)^{-p}\Big) \Big)^{\frac1p} \leq (1+\eps) \cdot \left( \tr\left(\mX^{-p}\right) \right)^{1/p}
\]
with probability at least $1-O\big(\frac{M^2}{\eps^2 p^2} \cdot e^{-\Omega\left( \gamma \sqrt{d} \right)} \big)$, 
where $\mX$ is an optimal fractional solution to \eqref{eq:convex-p}.
Moreover, the solution set $S$ satisfies
$|S| \leq (1+\eps)k + O\big( \frac{d}{\gamma} + \frac{d}{\kappa} \big)$
with probability at least $1-e^{-\Omega(\eps d / \min\{\gamma, \kappa\})}$. 
\end{theorem}

With \Cref{t:main-p-norm}, we can prove \Cref{t:main} by first noticing that both $\gamma$ and $\kappa$ are chosen in the order of $\Theta(\max\{ \eps, 1/p \})$ and then turning the bicriteria approximation result into a true approximation with a scaling argument as in~\cite{LZ21}. We omit the standard proof and refer the readers to~\cite{LZ21} for more details.

\subsection{The Proof Plan} \label{ss:plan}

In this subsection, we provide the overall plan and the precise statements for analyzing the randomized exchange algorithm, and then put together the statements to obtain the main technical theorem.

\subsubsection{Well-Defined Algorithm}

First, we prove that the randomized exchange algorithm is well-defined. In particular, we need to show that a fractional optimal solution to the convex relaxation~\eqref{eq:convex-p} with at most $O(d^2)$ fractional entries can be found in polynomial time, 
and also the probability distributions in the exchange subroutine are well-defined for $M = d^2 + \frac{d}{\gamma} + 1$. 
These can be established using the same arguments (with a different value for $\gamma$) as in Lemma~4.2 and Claim~4.4 in~\cite{LZ21}. 
Note that the modified exchange subroutine does not affect these arguments. 

The following simple observation (Observation 4.3 in~\cite{LZ21}) will be useful in the analysis of the algorithm.

\begin{observation} \label{obs:prob}
    For any $t \geq 0$, it holds that $i \in S_t$ for all $i$ with $\vx(i) = 1$ and $j \in [n] \backslash S_t$ for all $j$ with $\vx(j) = 0$. 
    This further implies that $\Pr(i_t = i) = 0$ for all $i$ with $\vx(i) \in \{0,1\}$ and $\Pr(j_t = j) = 0$ for all $j$ with $\vx(j) \in \{0,1\}$.
\end{observation}

\subsubsection{Solution Size Bound} 

Then, we show that the algorithm returns a solution set $S$ of size not much larger than $k$ with high probability.

\begin{theorem}[Variant of Theorem 3.12 of~\cite{LZ20}] \label{t:cost-kappa}
    Let $\alpha = \sqrt{d}/\gamma$ and $\kappa$ be the parameters used in the randomized exchange algorithm.
    Suppose that the solution $S_t$ of the randomized exchange algorithm satisfies $\lambda_{\min}\left(\sum_{i \in S_{t}} \vv_i \vv_i^\top\right) < 1$ for all $t \in [\tau]$. 
    Then, for any given $\delta \in [0,1]$,
    \[
    \Pr\bigg[ |S_{\tau}| \leq (1+\delta) \cdot \sum_{i=1}^n \vx(i) + \Big( \frac{12d}{\gamma} + \frac{2d}{\kappa} \Big) \bigg] \geq 1 - \exp\Big(-\Omega\Big(\frac{\delta d}{\min\{\gamma, \kappa\}}\Big)\Big).
    \]
\end{theorem}

When $\kappa=1$, the above theorem follows directly from the one-sided spectral rounding result in~\cite{LZ20}. 
With smaller $\kappa$, we are restricting the set of vectors that can be swapped out from the current solution. 
This would increase the chance of not removing a vector, 
and thus increasing the size of the solution. 
Fortunately, we can show that the increase of the solution size can be bounded by an additive $d/\kappa$ term. 
The proof idea is similar to the one in~\cite{LZ20}, and the main difference is a modified bound on the expected change of size. 
We provide a proof sketch in \Cref{a:cost} for completeness.

\subsubsection{Approximation Guarantee} 

The most technical part of the proof is to establish the approximation guarantee.
We follow the analysis in~\cite{LZ21} to conceptually divide the execution of the algorithm into two phases as described in \Cref{ss:previous}.

In the first phase, we show that the minimum eigenvalue will reach $1-2\gamma$ in $O(M/\gamma)$ iterations with high probability, which follows from the spectral rounding result in~\cite{LZ20}.
    
\begin{proposition}[Counterpart of Proposition 4.5 in~\cite{LZ21}] \label{p:terminate1-p}
The probability that the randomized exchange algorithm has terminated successfully within $2M/\gamma$ iterations or there exists $\tau_1 \leq 2M/\gamma$ with $\lambda_{\min}(\mZ_{\tau_1}) \geq 1-2\gamma$ is at least $1-\exp(-\Omega(\sqrt{d}))$. 
\end{proposition}

In the second phase, the minimum eigenvalue will be at least $1-5\gamma$ during the next $\Theta\left(\frac{M}{\eps p}\right)$ iterations with good probability.

\begin{proposition}[Counterpart of Proposition 4.6 in~\cite{LZ21}] \label{p:bounded_min-p}
    Suppose $\lambda_{\min}(\mZ_{\tau_1}) \geq 1-2\gamma$ for some $\tau_1$.
    In the randomized exchange algorithm, the probability that $\lambda_{\min}(\mZ_t) \geq 1-5\gamma$ for all $\tau_1 \leq t \leq \tau_1 + \frac{2M}{\eps p}$ is at least $1 -\frac{4M^2}{\eps^2 p^2} \cdot e^{- \Omega(\gamma \sqrt{d})}$. 
\end{proposition}
    
Both of the proofs of \Cref{p:terminate1-p} and \Cref{p:bounded_min-p} follow same arguments in~\cite{LZ21} (with a new parameter $\gamma$).
We remark that the modified exchange subroutine with a restricted $S'_t$ does not affect these arguments, as removing less vectors only helps to improve the minimum eigenvalue of the solution. 
We omit the proofs and refer the readers to~\cite{LZ21}.  
    
The main technical contribution in this paper is to prove that the $\Phi_p$ objective will improve to at most $(1+\eps)$ times the optimal value during the second phase when the minimum eigenvalue is at least $1-5\gamma$.

\begin{restatable}{theorem}{termination} \label{t:termination-p}
Given $\eps \in (0,1)$, if $p \leq 1/\eps$ and $k \gtrsim \frac{p d}{\eps}$ for some $\eps \in (0,1)$, 
    then the probability that the following three events happen simultaneously during the execution of the randomized exchange algorithm is at most $\exp(-\Omega( \gamma \sqrt{d}))$.
\begin{enumerate} \setlength\itemsep{0mm}
\item $\lambda_{\min}(\mZ_{\tau_1}) \geq 1-2\gamma$ for some $\tau_1$;
\item $\lambda_{\min}(\mZ_t) \geq 1-5\gamma$ for all $\tau_1 \leq t \leq \tau_1 + \frac{2M}{\eps p}$; 
\item the randomized exchange algorithm has not terminated by time $\tau_1 + \frac{2M}{\eps p}$.
\end{enumerate}
\end{restatable}

\subsubsection{Proof of \Cref{t:main-p-norm}} \label{ss:wrapup}

We put together the statements in this subsection to obtain \Cref{t:main-p-norm}.

\begin{proof}[Proof of \Cref{t:main-p-norm}]
We start with analyzing the approximation guarantee in the theorem.

Firstly, consider the easier case $p \geq 1/\eps$, which implies that $\gamma =\eps/6$.
By \Cref{p:terminate1-p}, there exists $\tau_1 \leq 2M/\gamma$ such that $\lambda_{\min}(\mZ_{\tau_1}) \geq (1-\eps/3)$ with probability $1 - \exp(-\Omega(\sqrt{d}))$. 
We note that $\lambda_{\min}(\mZ_{\tau_1}) \geq (1-\eps/3)$ is equivalent to $\mY_{\tau_1} \succcurlyeq (1-\eps/3) \mX$, which is sufficient to establish \[(\tr(\mY_{\tau_1}^{-p}))^{1/p} \leq (1+\eps) (\tr(\mX^{-p}))^{1/p},\]
i.e., the approximation guarantee in the theorem.
We remark that we do not need the assumption $k \gtrsim d/(\gamma \eps)$ in the proof of this case.

Then, we consider the case of $p \leq 1/\eps$ and define the bad events for the randomized exchange algorithm:
\begin{itemize}
\setlength\itemsep{-1pt}
    \item $B_1$: the algorithm has not terminated successfully within $2M/\gamma$ iterations and $\tau_1 > 2M/\gamma$ where $\tau_1$ is the first time such that $\lambda_{\min}(\mZ_{\tau_1}) \geq 1-2\gamma$.
    \item $B_2$: there exists some $\tau_1 \leq t \leq \tau_1 + \frac{2M}{\eps p}$ such that $\lambda_{\min}(\mZ_t) < 1-5\gamma$.
    \item $B_3$: the termination condition $ (\tr( \mY_t^{-p} ) )^{\frac1p} \leq (1+\eps) (\tr(\mX^{-p}))^{\frac1p} $  is not satisfied for all $\tau_1 \leq t \leq \tau_1 + \frac{2M}{\eps p}$.
\end{itemize}
If none of the bad events happens, then either the algorithm has terminated successfully within $2M/\gamma$ iterations or the termination condition will be satisfied at some time $t \leq \tau_1 + \frac{2M}{\eps p} \leq \frac{2M}{\gamma} + \frac{2M}{\eps p}$.
So, the probability that the randomized exchange algorithm has not satisfied the termination condition within $\frac{2M}{\gamma} + \frac{2M}{\eps p}$ iterations is upper bounded by
\begin{align*}
\Pr[B_1 \cup B_2 \cup B_3] 
& = \Pr[B_1] 
+ \Pr[B_2 \cap \neg B_1] 
+ \Pr[B_3 \cap \neg B_2 \cap \neg B_1]
\\
& \leq O\left(e^{-\Omega(\sqrt{d})} \right)
+ O\bigg( \frac{M^2}{\eps^2 p^2} \cdot e^{-\Omega(\gamma \sqrt{d})}  \bigg)
+ O\left( e^{-\Omega(\gamma \sqrt{d})}  \right)
\\
& \leq O\bigg( \frac{M^2}{\eps^2 p^2} \cdot e^{-\Omega(\gamma \sqrt{d})}  \bigg),
\end{align*}
where $\Pr[B_1]$ is bounded in Proposition~\ref{p:terminate1-p},
$\Pr[B_2 \cap \neg B_1]$ is bounded in Proposition~\ref{p:bounded_min-p},
and $\Pr[B_3 \cap \neg B_2 \cap \neg B_1]$ is bounded in Theorem~\ref{t:termination-p} (note that we need the assumption $p \leq 1/\eps$ and $k \gtrsim p d /\eps$ here).
The termination condition implies the approximation guarantee directly.

Finally, we consider the size of the returned solution.
Note that if $\lambda_{\min}(\mZ_t) \geq 1$ then $\mY_t \succcurlyeq \mX$, which further implies that the termination condition is met at time $t$.
Hence, we can assume $\lambda_{\min}(\mZ_t) < 1$ before the algorithm terminates.
Therefore, we can apply Theorem~\ref{t:cost-kappa} to conclude that the returned solution $S$ satisfies 
$
|S| \leq (1+\eps)k + O\big( \frac{d}{\gamma} + \frac{d}{\kappa} \big) 
$
with probability at least $1-\exp(-\Omega(\frac{\eps d}{\min\{\gamma, \kappa\}}))$.
\end{proof}

\subsection{New Ideas} \label{ss:ideas}

The key in proving \Cref{t:termination-p} is to bound the change of the objective value after an exchange.
For A-design ($p=1$), there is a simple inequality bounding the change of the objective as
\begin{align*}
    & \tr \Big(\big(\mY - \vv \vv^\top + \vw \vw^\top\big)^{-1}\Big) \leq \tr\big(\mY^{-1}\big) + \underbrace{\frac{ \vv^\top \mY^{-2} \vv}{1 - \inner{\vv \vv^\top}{\mY^{-1}}} - \frac{ \vw^\top \mY^{-2} \vw }{1 + \inner{\vw \vw^\top}{\mY^{-1}}}}_{\text{progress}}.
\end{align*}
For general $\Phi_p$-design, the change of the $\Phi_p$ function under rank-two updates is considerably more complicated.
Using Sherman-Morrison formula and Lieb-Thirring inequality, 
we can bound the change of the $\Phi_p$ objective (in fact, the $p$-th power of the $\Phi_p$ objective) as follows:
\begin{multline*}
\tr\left( (\mY + \vw \vw^\top - \vv \vv^\top)^{-p} \right) \\ 
\leq
\tr(\mY^{-p}) + \sum_{i=1}^p \binom{p}{i} \bigg( (-1)^i \frac{(\vw^\top \mY^{-1} \vw)^{i-1} \cdot \vw^\top \mY^{-p-1} \vw}{(1+\vw^\top \mY^{-1} \vw)^i} + \frac{(\vv^\top \mY^{-1} \vv)^{i-1} \cdot \vv^\top \mY^{-p-1} \vv}{(1-\vv^\top \mY^{-1} \vv)^i} \bigg).
\end{multline*}
There are many higher order terms introduced by the $p$-norm,
and dealing with these is the main technical difficulty in this paper.

As discussed in \Cref{r:exponential}, if we use the same algorithm in~\cite{LZ21}, with some careful manipulations including applying H\"older's inequality appropriately, we can achieve $(1+\eps)$-approximation but with the much worst requirement that $k \gtrsim 2^{O(p)} d /\eps$.
The reason is that removing some ``influential'' vectors (even with relatively small probability) from the current solution will blow up the expectation of the change of the objective function due to the higher order terms in the above inequality.

To overcome this issue, we introduce the parameter $\kappa$ and modify the randomized exchange algorithm by restricting those vectors (in $S_t'$) that are allowed to swap out of the current solution.
This helps us to effectively bound those higher order terms in the above inequality about the change of the objective function.
But, with smaller $\kappa$, we are restricting the set of vectors that can be swapped out from the current solution. 
This would increase the chance of not removing a vector, 
and thus increasing the size of the solution. 
Fortunately, we can show that the increase of the solution size can be bounded by an additive $d/\kappa$ term as described in \Cref{t:cost-kappa}.

The analysis for A-design in~\cite{LZ21} contains two parts: (1) bound the expected progress; (2) prove the concentration of the total progress. 
The condition that the minimum eigenvalue is at least $1/4$ is very important in both parts, and the optimality conditions for the convex program~\eqref{eq:convex-p} is crucially used in the concentration argument.
Interestingly, the optimality condition of the convex program~\eqref{eq:convex-p} is also crucial in bounding the expected progress for $\Phi_p$ objective with higher $p$. 
Much effort in this paper is used to get the expectations of the objective value right, while it was relatively easy for D/A-design (when $p=0,1$).

\section{The Analysis of $\Phi_p$ Objective} \label{s:objective}

The goal of this section is to prove \Cref{t:termination-p}. 
Since we are focusing on the case of $p \leq 1/\eps$ in \Cref{t:termination-p}, we can assume without loss of generality that
\[
\gamma = \max\left\{ \frac{\eps}{6}, \frac{1}{6p} \right\} = \frac{1}{6p} \qquad \text{and} \qquad \kappa = \max\left\{ \frac{\eps}{2}, \frac{1}{2p} \right\} = \frac{1}{2p}
\]
in the remaining of this section.
For ease of notation, we also assume the start time of the second phase is $\tau_1 = 1$. 
To analyze progress of the algorithm in terms of the objective function, we upper bound the change of $\tr(\mY_t^{-p})$ 
after a swap using the following lemma, for which we will provide a proof in \Cref{ss:one-step}.

\begin{lemma} \label{l:one-step-p}
Let $\mY \succ 0$ be a $d$-dimensional positive definite matrix and $p \geq 1$ be an integer. For any $\vw \in \R^d$ and $\vv \in \R^d$ such that $\vv^\top \mY^{-1} \vv < 1$, 
\begin{multline*}
\tr\Big( (\mY + \vw \vw^\top - \vv \vv^\top)^{-p} \Big) \\ 
\leq 
\tr(\mY^{-p}) + \sum_{i=1}^p \binom{p}{i} \bigg( (-1)^i \frac{(\vw^\top \mY^{-1} \vw)^{i-1} \cdot \vw^\top \mY^{-p-1} \vw}{(1+\vw^\top \mY^{-1} \vw)^i} + \frac{(\vv^\top \mY^{-1} \vv)^{i-1} \cdot \vv^\top \mY^{-p-1} \vv}{(1-\vv^\top \mY^{-1} \vv)^i} \bigg).
\end{multline*}
\end{lemma}

In the randomized exchange algorithm, we swap vectors $\vu_{i_t}$ and $\vu_{j_t}$ in each iteration where $\vu_{i_t}$ is in the current solution. 
Thus, $\vu_{i_t}^\top \mY_t^{-1} \vu_{i_t} \leq 1$ always holds, and in fact it will be clear later that $\vu_{i_t}^\top \mY_t^{-1} \vu_{i_t}$ is strictly less than 1. 
Hence, \Cref{l:one-step-p} can be applied repeatedly to obtain that, for any $\tau \geq 1$,
\begin{align} \label{eq:trace-bound-p}
    \tr(\mY_{\tau+1}^{-p}) - \tr(\mY_1^{-p})
    & \leq \sum_{t=1}^\tau \bigg( \underbrace{\sum_{i=1}^p \binom{p}{i} \frac{(\vu_{i_t}^\top \mY_t^{-1} \vu_{i_t})^{i-1} \cdot \vu_{i_t}^\top \mY_t^{-p-1} \vu_{i_t}}{(1-\vu_{i_t}^\top \mY_t^{-1} \vu_{i_t})^i}}_{\rm loss} \\
    & \qquad \qquad - \underbrace{\sum_{i=1}^p \binom{p}{i}  (-1)^{i+1} \frac{(\vu_{j_t}^\top \mY_t^{-1} \vu_{j_t})^{i-1} \cdot \vu_{j_t}^\top \mY_t^{-p-1} \vu_{j_t}}{(1+\vu_{j_t}^\top \mY_t^{-1} \vu_{j_t})^i}}_{\rm gain} \bigg). \nonumber
\end{align}

We define gain $g_t$, loss $l_t$, and progress $\Gamma_t$ in the $t$-th iteration as follows
\begin{align*}
& g_t := \sum_{i=1}^p \binom{p}{i}  (-1)^{i+1} \frac{(\vu_{j_t}^\top \mY_t^{-1} \vu_{j_t})^{i-1} \cdot \vu_{j_t}^\top \mY_t^{-p-1} \vu_{j_t}}{(1+\vu_{j_t}^\top \mY_t^{-1} \vu_{j_t})^i}, \\
& l_t := \sum_{i=1}^p \binom{p}{i} \frac{(\vu_{i_t}^\top \mY_t^{-1} \vu_{i_t})^{i-1} \cdot \vu_{i_t}^\top \mY_t^{-p-1} \vu_{i_t}}{(1-\vu_{i_t}^\top \mY_t^{-1} \vu_{i_t})^i}, \\ 
& \Gamma_t := g_t - l_t.
\end{align*}

In Section~\ref{ss:expectation}, we will prove that if $\vx$ is a fractional optimal solution to \eqref{eq:convex-p} and $\mZ_t$ has lower-bounded minimum eigenvalue and the objective value of the current solution is far from optimal, then the expected progress in the $t$-th iteration is large. 
Then, in Section~\ref{ss:concentration}, we will prove that the total progress is concentrated around its expectation. 
Finally, we complete the proof of Theorem~\ref{t:termination-p} in Section~\ref{ss:termination-p}.

\subsection{Change of Objective Value in One Step} \label{ss:one-step}

In~\cite{LZ21}, a rank-two update formula is used to compute the change of the objective value in one step when $p=1$.
For general $\Phi_p$-design, the rank-two update formula becomes considerably more complicated, and instead we do the update in two smaller steps:
We first use a rank-one update to add $\vu_{j_t}$ to the current solution, then use another rank-one update to remove $\vu_{i_t}$ from the current solution.

\begin{proof}[Proof of \Cref{l:one-step-p}]
Let $\mY_1 = \mY + \vw \vw^\top$. By Sherman-Morrison formula (\Cref{l:Sherman-Morrison}), it holds that
\[
\tr(\mY_1^{-p}) = \tr\left( \Big( \mY^{-1} - \frac{\mY^{-1} \vw \vw^\top \mY^{-1}}{1 + \vw^\top \mY^{-1} \vw} \Big)^p \right) = \tr\left( \Big( \mY^{-1/2} \Big( \mI - \frac{\mY^{-1/2} \vw \vw^\top \mY^{-1/2}}{1 + \vw^\top \mY^{-1} \vw }\Big) \mY^{-1/2} \Big)^p \right).
\]
Then, we can apply Lieb-Thirring inequality (\Cref{l:Lieb-Thirring}) to show that
\[
\tr(\mY_1^{-p}) \leq \tr\left( \mY^{-p/2} \Big( \mI - \frac{\mY^{-1/2} \vw \vw^\top \mY^{-1/2}}{1 + \vw^\top \mY^{-1} \vw }\Big)^p \mY^{-p/2} \right) = \tr\left( \mY^{-p} \Big( \mI - \frac{\mY^{-1/2} \vw \vw^\top \mY^{-1/2}}{1 + \vw^\top \mY^{-1} \vw }\Big)^p \right).
\]
Expanding by the binomial theorem, 
\begin{align}
\tr(\mY_1^{-p}) & \leq \sum_{i=0}^p (-1)^i \binom{p}{i} \tr\left( \mY^{-p} \Big( \frac{\mY^{-1/2} \vw \vw^\top \mY^{-1/2}}{1 + \vw^\top \mY^{-1} \vw} \Big)^i\right) \nonumber \\
& = \tr(\mY^{-p}) + \sum_{i=1}^p (-1)^i \binom{p}{i} \frac{(\vw^\top \mY^{-1} \vw)^{i-1} \cdot \vw^\top \mY^{-p-1} \vw}{(1+\vw^\top \mY^{-1} \vw)^i}. \label{eq:add}
\end{align}

For $\mY_2 = \mY_1 - \vv \vv^\top$, we can apply similar argument to show that
\[
\tr(\mY_2^{-p}) \leq \tr(\mY_1^{-p}) + \sum_{i=1}^p \binom{p}{i} \frac{(\vv^\top \mY_1^{-1} \vv)^{i-1} \cdot \vv^\top \mY_1^{-p-1} \vv}{(1-\vv^\top \mY_1^{-1} \vv)^i}.
\]
Notice that $\mY_1 = \mY + \vw \vw^\top \succcurlyeq \mY$ and $\vv^\top \mY^{-1} \vv < 1$, thus it holds that
\begin{align} \label{eq:remove}
\tr(\mY_2^{-p}) \leq \tr(\mY_1^{-p}) + \sum_{i=1}^p \binom{p}{i} \frac{(\vv^\top \mY^{-1} \vv)^{i-1} \cdot \vv^\top \mY^{-p-1} \vv}{(1-\vv^\top \mY^{-1} \vv)^i}.
\end{align}
The lemma follows by combining \eqref{eq:add} and \eqref{eq:remove}.
\end{proof}

\begin{remark} \label{r:real_p}
Lemma~\ref{l:one-step-p} can be generalized to all real $p \geq 1$ by invoking Newton's generalized binomial theorem in the proof. 
To guarantee the convergence of the generalized binomial theorem, we need to ensure a stronger condition $\vv^\top \mY^{-1} \vv \leq \frac12$.
This is not an issue for our application, as our algorithm always removes vectors from the restricted set $S'_t$, which guarantees that $\vv^\top \mY^{-1} \vv \leq \frac12$ is satisfied.
Given the new version of Lemma~\ref{l:one-step-p} with real $p$, we can generalize the main result in this paper (i.e., Theorem~\ref{t:main}) to all real $p \geq 1$ with essentially the same analysis.
\end{remark}

\subsection{Expected Progress} \label{ss:expectation}

To analyze the expected progress, we need to use the following two lemmas. The first one is an implication of the lower-bounded minimum eigenvalue condition, which is an analog of Lemma 4.13 in~\cite{LZ21}. We provide a proof in \Cref{a:proofs-sec-term} for completeness.

\begin{lemma} \label{l:E-A-gamma}
For $\gamma \leq \frac16$, 
if $\mZ_t \succcurlyeq (1-5\gamma) \mI$, then 
\begin{align*}
    \inner{\vv_i \vv_i^\top}{\mZ_t^{-1}} \leq \alpha \inner{\vv_i \vv_i^\top}{\mA_t^{\frac12}} \leq \alpha \lambda_{\min}(\mZ_t) \inner{\vv_i \vv_i^\top}{\mZ_t^{-1}} \qquad \forall i \in [n].
\end{align*}
\end{lemma}

The other one is an implication of the optimality condition of \eqref{eq:convex-p}.
The proof of the lemma is similar to the one in~\cite{LZ21} for A-design, we include it in \Cref{a:proofs-sec-term} for completeness. 

\begin{restatable}{lemma}{optp} \label{l:optimality-p}
Let $\vx \in [0,1]^n$ be an optimal fractional solution of the convex programming relaxation~\eqref{eq:convex-p} for the $p$-norm problem.
Then, for each $1 \leq i \leq n$ with $0 < \vx(i) < 1$,  
\[
\inner{\mX^{-p-1}}{\vu_i \vu_i^\top} \leq \frac{1}{k} \cdot \tr(\mX^{-p}).
\]
\end{restatable}

Now, we are ready to lower bound the expected progress. 
We will first handle the expected loss and expected gain separately in \Cref{l:loss-p} and \Cref{l:gain-p}. Then, combine the two parts to lower bound the expected progress in \Cref{l:p-exp-progress}.
For simplicity, we denote $\E_t[ \cdot ]$ as the conditional expectation given what had happened up to time $t$, that is, $\E[ \cdot \mid S_{t-1}]$.

\subsubsection{Expected Loss}

The minimum eigenvalue lower bound (\Cref{l:E-A-gamma}), the optimality condition (\Cref{l:optimality-p}), and the introduction of the new parameter $\kappa$ in the randomized exchange algorithm are all crucial in the following lemma.

\begin{lemma}[Expected Loss] \label{l:loss-p}
Let $S_{t-1}$ be the solution set at time $t$ and $\mZ_t = \sum_{i \in S_{t-1}} \vv_i \vv_i^\top$ for $1 \leq t \leq \tau$.
Suppose $\vx$ is an optimal solution of \eqref{eq:convex-p}, $\lambda_{\min}(\mZ_t) \in [1-5\gamma,1)$, $\gamma = 1/6p$, and $\kappa = 1/2p$.
Then
\[
\E_t[l_t] \leq \frac{p}{M} \Big( \tr(\mY_t^{-p}) - \inner{\mX_{S_{t-1}}}{\mY_t^{-p-1}} \Big) + O\Big(\frac{p^2 d}{kM} \Big) \cdot \tr(\mX^{-p}), 
\]
where we denote $\mX_S := \sum_{i \in S} \vx(i) \vu_i \vu_i^\top$ for any set $S \subseteq [n]$. 
\end{lemma}

\begin{proof}
There are $p$ terms in the loss term $l_t$. We deal with the linear term and higher order terms separately. 
Consider the linear term:
\begin{align*}
    \E_t\left[ \frac{p \cdot \vu_{i_t}^\top \mY_t^{-p-1} \vu_{i_t} }{1- \vu_{i_t}^\top \mY_t^{-1} \vu_{i_t}} \right] & = \sum_{i \in S'_{t-1}} \frac{1- \vx(i)}{M} \Big( 1 - \alpha \inner{\vv_i \vv_i^\top}{\mA_t^{1/2}} \Big) \cdot \frac{p \cdot \vu_i^\top \mY_t^{-p-1} \vu_i }{1- \vu_i^\top \mY_t^{-1} \vu_i} \\
    & = \sum_{i \in S'_{t-1}} \frac{1- \vx(i)}{M} \Big( 1 - \alpha \inner{\vv_i \vv_i^\top}{\mA_t^{1/2}} \Big) \cdot \frac{p \cdot \vu_i^\top \mY_t^{-p-1} \vu_i }{1- \inner{\vv_i \vv_i^\top}{\mZ_t^{-1}}},
\end{align*}
where the second line follows by the definitions of $\mY_t$ and $\mZ_t$, which implies that
\begin{align} \label{eq:Y-Z}
    \inner{\vv_i \vv_i^\top}{\mZ_t^{-1}} = \vu_i^\top \mY_t^{-1} \vu_i.
\end{align}

Note that $\gamma = 1/6p \leq 1/6$. Thus, we can apply the first inequality in \Cref{l:E-A-gamma} and then relax $S'_{t-1}$ to $S_{t-1}$ to obtain that
\begin{align} \label{eq:l-linear}
    \E_t\left[ \frac{p \cdot \vu_{i_t}^\top \mY_t^{-p-1} \vu_{i_t} }{1- \vu_{i_t}^\top \mY_t^{-1} \vu_{i_t}} \right] \leq \frac{p}{M} \sum_{i \in S_{t-1}} (1-\vx(i)) \cdot \vu_i^\top \mY_t^{-p-1} \vu_i = \frac{p}{M} \Big( \tr(\mY_t^{-p}) - \inner{\mX_{S_{t-1}}}{\mY_t^{-p-1}} \Big).
\end{align}

Then, we consider the remaining $p-1$ higher order loss terms.
\begin{multline*}
    \underbrace{\E_t\left[ \sum_{l=2}^p \binom{p}{l} \frac{(\vu_{i_t}^\top \mY_t^{-1} \vu_{i_t})^{l-1} \cdot \vu_{i_t}^\top \mY_t^{-p-1} \vu_{i_t}}{(1-\vu_{i_t}^\top \mY_t^{-1} \vu_{i_t})^l} \right]}_{\small \Circled{1}} \\
    = \sum_{i \in S'_{t-1}} \frac{1- \vx(i)}{M} \Big( 1 - \alpha \inner{\vv_i \vv_i^\top}{\mA_t^{1/2}} \Big) \sum_{l=2}^p \binom{p}{l} \frac{(\vu_i^\top \mY_t^{-1} \vu_i)^{l-1} \cdot \vu_i^\top \mY_t^{-p-1} \vu_i}{(1-\vu_i^\top \mY_t^{-1} \vu_i)^l}.
\end{multline*}

Again, by applying \Cref{l:E-A-gamma}, it holds that
\begin{align*}
    \Circled{1} \leq & \sum_{i \in S'_{t-1}} \frac{1- \vx(i)}{M} \cdot \vu_i^\top \mY_t^{-p-1} \vu_i \cdot \sum_{l=2}^p \binom{p}{l} \frac{(\vu_i^\top \mY_t^{-1} \vu_i)^{l-1} }{(1-\vu_i^\top \mY_t^{-1} \vu_i)^{l-1}}.
\end{align*}

Notice that
\begin{align} \label{eq:Y-X}
    \lambda_{\min}(\mZ_t) \geq 1-5\gamma \qquad \iff \qquad \mY_t \succcurlyeq (1-5\gamma) \mX.
\end{align}
Thus, by the assumption $\lambda_{\min}(\mZ_t) \geq 1-5\gamma$, it follows that
\begin{align*}
    \Circled{1} \leq & \frac{(1-5\gamma)^{-p-1}}{M} \sum_{i \in S'_{t-1}} (1- \vx(i)) \cdot \vu_i^\top \mX^{-p-1} \vu_i \cdot \sum_{l=2}^p \binom{p}{l} \frac{(\vu_i^\top \mY_t^{-1} \vu_i)^{l-1} }{(1-\vu_i^\top \mY_t^{-1} \vu_i)^{l-1}}.
\end{align*}
Using the fact that $\vx$ is a fractional optimal solution to \eqref{eq:convex-p} and then applying \Cref{l:optimality-p}, it holds that
\begin{align*}
    \Circled{1} & \leq \frac{(1-5\gamma)^{-p-1}}{kM} \cdot \tr(\mX^{-p}) \sum_{i \in S'_{t-1}} (1- \vx(i)) \cdot \sum_{l=2}^p \binom{p}{l} \frac{(\vu_i^\top \mY_t^{-1} \vu_i)^{l-1} }{(1-\vu_i^\top \mY_t^{-1} \vu_i)^{l-1}} \\
    & \leq \frac{(1-5\gamma)^{-p-1}}{kM} \cdot \tr(\mX^{-p}) \sum_{l=2}^p \binom{p}{l} \sum_{i \in S'_{t-1}} \vu_i^\top \mY_t^{-1} \vu_i \cdot \frac{(\vu_i^\top \mY_t^{-1} \vu_i)^{l-2} }{(1-\vu_i^\top \mY_t^{-1} \vu_i)^{l-1}}.
\end{align*}

Due to the definition of the set $S'_{t-1}$ and \eqref{eq:Y-Z}, it holds that $\vu_i^\top \mY_t^{-1} \vu_i \leq \kappa$ for all $i \in S'_{t-1}$. Thus,
\begin{align*}
    \Circled{1} \leq \frac{(1-5\gamma)^{-p-1}}{kM} \cdot \tr(\mX^{-p}) \sum_{l=2}^p \binom{p}{l} \sum_{i \in S'_{t-1}} \vu_i^\top \mY_t^{-1} \vu_i \cdot \Big(\frac{\kappa}{1-\kappa}\Big)^{l-2} \cdot \frac{1}{1-\kappa}.
\end{align*}
Since $\sum_{i \in S'_{t-1}} \vu_i^\top \mY_t^{-1} \vu_i \leq \inner{Y_t}{Y_t^{-1}} = d$, we can further upper bound the $\Circled{1}$ term by
\begin{align}
    \Circled{1} ~\leq~ & \frac{(1-5\gamma)^{-p-1} d}{kM} \cdot \tr(\mX^{-p}) \cdot \sum_{l=2}^p \binom{p}{l} \Big(\frac{\kappa}{1-\kappa}\Big)^l \cdot \frac{1-\kappa}{\kappa^2} \nonumber \\
    =~ & \frac{(1-5\gamma)^{-p-1} d}{kM} \cdot \tr(\mX^{-p}) \cdot \frac{1-\kappa}{\kappa^2} \Big( \frac{1}{(1-\kappa)^p} - 1 - \frac{p \kappa}{1-\kappa} \Big) \nonumber \\
    \leq~ & \frac{(1-5\gamma)^{-p-1} d}{kM} \cdot \tr(\mX^{-p}) \cdot \frac{1}{\kappa^2} \cdot \frac{1}{(1-\kappa)^p}. \label{eq:l-higher}
\end{align}

Combining \eqref{eq:l-linear} and \eqref{eq:l-higher}, the expected loss can be bounded by
\begin{align*} 
    \E_t[l_t] & \leq \frac{p}{M} \Big( \tr(\mY_t^{-p}) - \inner{\mX_{S_{t-1}}}{\mY_t^{-p-1}} \Big) + \frac{(1-5\gamma)^{-p-1} d}{\kappa^2 (1-\kappa)^p kM} \cdot \tr(\mX^{-p}) \\
    & \leq \frac{p}{M} \Big( \tr(\mY_t^{-p}) - \inner{\mX_{S_{t-1}}}{\mY_t^{-p-1}} \Big) + O\Big(\frac{p^2 d}{kM}\Big) \cdot \tr(\mX^{-p}),
\end{align*}
where the last inequality follows by the choice of $\gamma$ and $\kappa$.
\end{proof}

\subsubsection{Expected Gain}

The analysis of the expected gain is slightly more complicated since the sampling probability does not cancel out with the denominator term in $g_t$ as nicely as in the analysis of loss $l_t$, and so we will divide into two cases and use different arguments. 
Again the minimum eigenvalue lower bound (\Cref{l:E-A-gamma}) and the optimality condition (\Cref{l:optimality-p}) are crucial in the following lemma.

\begin{lemma}[Expected Gain] \label{l:gain-p}
Let $S_{t-1}$ be the solution set at time $t$ and $\mZ_t = \sum_{i \in S_{t-1}} \vv_i \vv_i^\top$ for $1 \leq t \leq \tau$.
Suppose $\vx$ is an optimal solution to \eqref{eq:convex-p}, $\lambda_{\min}(\mZ_t) \in [1-5\gamma,1)$, and $\gamma = 1/6p$.
Then
\[
\E_t[g_t] \geq \frac{p}{M} \Big( \inner{\mX}{\mY_t^{-p-1}} - \inner{\mX_{S_{t-1}}}{\mY_t^{-p-1}} \Big) - O\Big( \frac{p^2d}{kM} \Big) \cdot \tr(\mX^{-p}).
\]
\end{lemma}
\begin{proof}
We consider two separate cases:
\begin{align*}
& S_1 := \{ j \in [m] \backslash S_{t-1} \mid \alpha \inner{\vv_j \vv_j^\top}{\mA_t^{1/2}} \geq p \cdot \vu_j^\top \mY_t^{-1} \vu_j \} \quad \text{and} \\
& S_2 := \{ j \in [m] \backslash S_{t-1} \mid \alpha \inner{\vv_j \vv_j^\top}{\mA_t^{1/2}} < p \cdot \vu_j^\top \mY_t^{-1} \vu_j \}.
\end{align*}
Then, the expected gain can be written as
\begin{align}
    \E_t[g_t] = & \sum_{j \in S_1} \frac{\vx(j)}{M} \cdot \Big( 1 + \alpha \inner{\vv_j \vv_j^\top}{\mA_t^{1/2}} \Big) \cdot \sum_{l=1}^p \binom{p}{l}  (-1)^{l+1} \frac{(\vu_j^\top \mY_t^{-1} \vu_j)^{l-1} \cdot \vu_j^\top \mY_t^{-p-1} \vu_j}{(1+\vu_j^\top \mY_t^{-1} \vu_j)^l} \label{eq:p-S1} \\
    & \qquad + \sum_{j \in S_2} \frac{\vx(j)}{M} \cdot \Big( 1 + \alpha \inner{\vv_j \vv_j^\top}{\mA_t^{1/2}} \Big) \cdot \sum_{l=1}^p \binom{p}{l}  (-1)^{l+1} \frac{(\vu_j^\top \mY_t^{-1} \vu_j)^{l-1} \cdot \vu_j^\top \mY_t^{-p-1} \vu_j}{(1+\vu_j^\top \mY_t^{-1} \vu_j)^l}. \label{eq:p-S2}
\end{align}

We analyze~\eqref{eq:p-S1} and~\eqref{eq:p-S2} separately. First, we consider~\eqref{eq:p-S1} and rearrange the $g_t$ term a bit, which will be easier for the analysis of \eqref{eq:p-S1}.
\begin{align} 
g_t & = \frac{u_{j_t}^\top \mY_t^{-p-1} \vu_{j_t}}{\vu_{j_t}^\top \mY_t^{-1} \vu_{j_t}} \cdot \sum_{i=1}^p \binom{p}{i}  (-1)^{i+1} \bigg( \frac{\vu_{j_t}^\top \mY_t^{-1} \vu_{j_t}}{1+\vu_{j_t}^\top \mY_t^{-1} \vu_{j_t}} \bigg)^i \nonumber \\
& = \frac{u_{j_t}^\top \mY_t^{-p-1} \vu_{j_t}}{\vu_{j_t}^\top \mY_t^{-1} \vu_{j_t}} \bigg( 1 - \bigg( 1 - \frac{\vu_{j_t}^\top \mY_t^{-1} \vu_{j_t}}{1 + \vu_{j_t}^\top \mY_t^{-1} \vu_{j_t}}\bigg)^p \bigg) \label{eq:gt1} \\
& = \frac{u_{j_t}^\top \mY_t^{-p-1} \vu_{j_t}}{\vu_{j_t}^\top \mY_t^{-1} \vu_{j_t}} \bigg( 1 - \frac{1}{(1 + \vu_{j_t}^\top \mY_t^{-1} \vu_{j_t})^p} \bigg). \label{eq:gt}
\end{align}
Thus, we can rewrite \eqref{eq:p-S1} as
\begin{align*}
    \eqref{eq:p-S1} = \sum_{j \in S_1} \frac{\vx(j)}{M} \cdot \Big( 1 + \alpha \inner{\vv_j \vv_j^\top}{\mA_t^{1/2}} \Big) \cdot \frac{u_j^\top \mY_t^{-p-1} \vu_j}{\vu_j^\top \mY_t^{-1} \vu_j} \bigg( 1 - \frac{1}{(1 + \vu_j^\top \mY_t^{-1} \vu_j)^p} \bigg).
\end{align*}
By the definition of $S_1$, it holds that
\begin{align*}
    \eqref{eq:p-S1} 
    & \geq \sum_{j \in S_1} \frac{\vx(j)}{M} \cdot u_j^\top \mY_t^{-p-1} \vu_j \cdot \frac{1 + p \cdot \vu_j^\top \mY_t^{-1} \vu_j}{\vu_j^\top \mY_t^{-1} \vu_j} \bigg( 1 - \frac{1}{(1 + \vu_j^\top \mY_t^{-1} \vu_j)^p} \bigg).
\end{align*}

Let $x = \vu_j^\top \mY_t^{-1} \vu_j > 0$ (as $\mY_t \succ 0$). Then, it holds that
\[
\frac{1+px}{x} \cdot \Big(1 - \frac{1}{(1+x)^p} \Big) = \frac{1+px}{x} \cdot \frac{(1+x)^p - 1}{(1+x)^p} = \frac{px(1+x)^p + (1+x)^p - 1 - px}{x(1+x)^p} \geq p.
\]
Thus,
\begin{align} \label{eq:S1}
    \eqref{eq:p-S1} & \geq \frac{p}{M} \sum_{j \in S_1} \vx(j) \cdot \vu_j^\top \mY_t^{-p-1} \vu_j.
\end{align}

Then, we consider~\eqref{eq:p-S2}. As in the proof of \Cref{l:loss-p}, we separate \eqref{eq:p-S2} into two parts, $\Circled{2}$ concerning the linear term and $\Circled{3}$ concerning the remaining $p-1$ higher order terms.
\begin{align*}
    \eqref{eq:p-S2} & = \underbrace{\sum_{j \in S_2} \frac{\vx(j)}{M} \cdot \Big( 1 + \alpha \inner{\vv_j \vv_j^\top}{\mA_t^{1/2}} \Big) \cdot \frac{p \cdot \vu_j^\top \mY_t^{-p-1} \vu_j}{1 + \vu_j^\top \mY_t^{-1} \vu_j}}_{\Circled{2}} \\
    & \qquad - \underbrace{\sum_{j \in S_2} \frac{\vx(j)}{M} \cdot \Big( 1 + \alpha \inner{\vv_j \vv_j^\top}{\mA_t^{1/2}} \Big) \cdot \sum_{l=2}^p \binom{p}{l}  (-1)^l \frac{(\vu_j^\top \mY_t^{-1} \vu_j)^{l-1} \cdot \vu_j^\top \mY_t^{-p-1} \vu_j}{(1+\vu_j^\top \mY_t^{-1} \vu_j)^l}}_{\Circled{3}}.
\end{align*}

The linear term $\Circled{2}$ is easy to bound, we can control it by \Cref{l:E-A-gamma} (combined with \eqref{eq:Y-Z}) so that
\begin{align} \label{eq:S2-linear}
    \Circled{2} \geq \frac{p}{M} \sum_{j \in S_2} \vx(j) \cdot \vu_j^\top \mY_t^{-p-1} \vu_j.
\end{align}

Then, we upper bound the higher order terms $\Circled{3}$ (notice that $\Circled{3}$ does not contain the minus sign). 
To upper bound $\Circled{3}$, for each $j \in S_2$, we can assume $\sum_{l=2}^p \binom{p}{l} (-1)^l \frac{(\vu_j^\top \mY_t^{-1} \vu_j)^{l-1} \cdot \vu_j^\top \mY_t^{-p-1} \vu_j}{(1+\vu_j^\top \mY_t^{-1} \vu_j)^l} \geq 0$ without loss of generality, as otherwise we can simply ignore the $j$-th term.
With this assumption and by the definition of $S_2$, it follows that
\begin{align*}
    \Circled{3} & \leq \frac1M \sum_{j \in S_2} \vx(j) (1+ p \vu_j^\top \mY_t^{-1} \vu_j) \sum_{l=2}^p \binom{p}{l}  (-1)^l \frac{(\vu_j^\top \mY_t^{-1} \vu_j)^{l-1} \cdot \vu_j^\top \mY_t^{-p-1} \vu_j}{(1+\vu_j^\top \mY_t^{-1} \vu_j)^{l}} \\
    & = \frac1M \sum_{j \in S_2} \vx(j) \vu_j^\top \mY_t^{-p-1} \vu_j \cdot (1+ p \vu_j^\top \mY_t^{-1} \vu_j) \sum_{l=2}^p \binom{p}{l}  (-1)^l \frac{(\vu_j^\top \mY_t^{-1} \vu_j)^{l-1}}{(1+\vu_j^\top \mY_t^{-1} \vu_j)^{l}}.
\end{align*}

Using the assumption $\lambda_{\min}(\mZ_t) \geq 1-5\gamma$ and \eqref{eq:Y-X}, it holds that
\begin{align*}
    \Circled{3} & \leq \frac{(1-5\gamma)^{-p-1}}{M} \sum_{j \in S_2} \vx(j) \vu_j^\top \mX^{-p-1} \vu_j \cdot (1+ p \vu_j^\top \mY_t^{-1} \vu_j) \sum_{l=2}^p \binom{p}{l}  (-1)^l \frac{(\vu_j^\top \mY_t^{-1} \vu_j)^{l-1}}{(1+\vu_j^\top \mY_t^{-1}. \vu_j)^{l}}.
\end{align*}

Since $\vx$ is an optimal solution to \eqref{eq:convex-p}, we can apply \Cref{l:optimality-p} and derive that
\begin{align*}
    \Circled{3} & \leq \frac{(1-5\gamma)^{-p-1}}{kM} \cdot \tr(\mX^{-p}) \sum_{j \in S_2} \vx(j) \cdot (1+ p \vu_j^\top \mY_t^{-1} \vu_j) \sum_{l=2}^p \binom{p}{l}  (-1)^l \frac{(\vu_j^\top \mY_t^{-1} \vu_j)^{l-1}}{(1+\vu_j^\top \mY_t^{-1} \vu_j)^{l}} \\
    & = \frac{(1-5\gamma)^{-p-1}}{kM} \cdot \tr(\mX^{-p}) \sum_{j \in S_2} \vx(j) \cdot \frac{1+ p \vu_j^\top \mY_t^{-1} \vu_j}{\vu_j^\top \mY_t^{-1} \vu_j} \bigg( \Big( 1 - \frac{\vu_j^\top \mY_t^{-1} \vu_j}{1 + \vu_j^\top \mY_t^{-1} \vu_j}\Big)^p - 1 + \frac{p \vu_j^\top \mY_t^{-1} \vu_j}{1+ \vu_j^\top \mY_t^{-1} \vu_j}\bigg).
\end{align*}

Let $x = \vu_j^\top \mY_t^{-1} \vu_j$, we want to upper bound
\[
 \frac{1+px}{x} \bigg( \Big( 1- \frac{x}{1+x}\Big)^p - 1 + \frac{px}{1+x}\bigg) = \frac{1}{x}\bigg( \Big( 1- \frac{x}{1+x}\Big)^p - 1 + \frac{px}{1+x}\bigg) + p\bigg( \Big( 1- \frac{x}{1+x}\Big)^p - 1 + \frac{px}{1+x}\bigg).
\]
For any $y \in [0,1]$, it holds that $(1-y)^p \leq 1- py + \binom{p}{2} y^2$. Thus, it follows that
\begin{align*}
& \frac{1}{x}\bigg( \Big( 1- \frac{x}{1+x}\Big)^p - 1 + \frac{px}{1+x}\bigg) + p\bigg( \Big( 1- \frac{x}{1+x}\Big)^p - 1 + \frac{px}{1+x}\bigg) \\
\leq~ & \frac1x \cdot \binom{p}{2} \cdot \frac{x^2}{(1+x)^2} + p \cdot \frac{px}{1+x} = \binom{p}{2} \cdot \frac{x}{1+x} + \frac{p^2 x}{1+x} \leq 2p^2 x.
\end{align*}

Therefore, the $\Circled{3}$ term can be further bounded by
\begin{align*}
\Circled{3} & \leq \frac{2 (1-5\gamma)^{-p-1} \cdot p^2}{kM} \cdot \tr(\mX^{-p}) \sum_{j \in S_2} \vx(j) \cdot \vu_j^\top \mY_t^{-1} \vu_j \leq \frac{2 (1-5\gamma)^{-p-1} \cdot p^2}{kM} \cdot \tr(\mX^{-p}) \inner{\mX}{\mY_t^{-1}}.
\end{align*}
Using \eqref{eq:Y-X}, it holds that
\begin{align} \label{eq:S2-higher}
\Circled{3} \leq \frac{2 (1-5\gamma)^{-p-2} \cdot p^2 d}{kM} \cdot \tr(\mX^{-p}) \leq O\Big( \frac{p^2 d}{kM} \Big) \cdot \tr(\mX^{-p}),
\end{align}
where the last inequality follows by the choice of $\gamma$.

Combining \eqref{eq:S1}, \eqref{eq:S2-linear} and \eqref{eq:S2-higher}, we can lower bound the expected gain by
\begin{align*}
\E_t[g_t] & \geq \frac{p}{M} \bigg( \sum_{j \in [m] \backslash S_{t-1}} \vx(j) \vu_j^\top \mY_t^{-p-1} \vu_j \bigg) - O\Big( \frac{p^2 d}{kM} \Big) \cdot \tr(\mX^{-p}) \\ 
& = \frac{p}{M} \Big( \inner{\mX}{\mY_t^{-p-1}} - \inner{\mX_{S_{t-1}}}{\mY_t^{-p-1}} \Big) - O\Big( \frac{p^2 d}{kM} \Big) \cdot \tr(\mX^{-p}). \qedhere  
\end{align*}
\end{proof}

\subsubsection{Expected Progress}
 
Finally, we apply H\"older's inequality appropriately to compare the gain term and the loss term.

\begin{lemma} \label{l:holder}
Given positive definite matrices $\mA, \mB \in \mathbb{S}^d_{++}$ and an integer $p \geq 1$, it holds that
\[
\inner{\mA}{\mB^{-p-1}} \geq \left( \frac{\tr(\mB^{-p})}{\tr(\mA^{-p})} \right)^{1/p} \cdot \tr(\mB^{-p}).
\]
\end{lemma}
\begin{proof}
Let $\mA = \sum_{i=1}^d a_i \vv_i \vv_i^\top$ be the eigendecomposition of $\mA$, and $\mB = \sum_{i=1}^d b_i \vw_i \vw_i^\top$ be the eigendecomposition of $\mB$.
Then,
\begin{align*}
\tr(\mB^{-p}) = \sum_{i=1}^d \frac{1}{b_i^p} = \sum_{1 \leq i,j \leq d}  \frac{1}{b_i^p} \inner{\vv_i}{\vw_j}^2 & = \sum_{1 \leq i,j \leq d} \frac{a_j^{p/(p+1)}}{b_i^p} \inner{\vv_i}{\vw_j}^{2p/(p+1)} \cdot \frac{1}{a_j^{p/(p+1)}} \inner{\vv_i}{\vw_j}^{2/(p+1)} \\
& \leq \bigg( \sum_{1 \leq i,j\leq d} \frac{a_j}{b_i^{p+1}} \inner{\vv_i}{\vw_j}^2 \bigg)^{p/(p+1)} \cdot \bigg( \sum_{1 \leq i,j\leq d} \frac{1}{a_j^p} \inner{\vv_i}{\vw_j}^2 \bigg)^{1/(p+1)} \\
& = \inner{\mA}{\mB^{-p-1}}^{p/(p+1)} \cdot \tr(\mA^{-p})^{1/(p+1)},
\end{align*}
where the inequality follows by H\"{o}lder's inequality.
Then, the lemma follows by taking the $(p+1)/p$'s power of both sides and rearranging the terms.
\end{proof}
 
Now, we lower bound the expected progress by combining \Cref{l:loss-p} and \Cref{l:gain-p}.
\begin{lemma} \label{l:p-exp-progress}
Let $S_{t-1}$ be the solution set at time $t$ and $\mZ_t = \sum_{i \in S_{t-1}} \vv_i \vv_i^\top$ for $1 \leq t \leq \tau$.
Suppose $\vx$ is an fractional optimal solution to \eqref{eq:convex-p}, and $\lambda_{\min}(\mZ_t) \in [1-5\gamma,1)$, $(\tr(\mY_t^{-2}))^{1/p} \geq \lambda \cdot (\tr\left( \mX^{-2} \right))^{1/p}$ for some $\lambda \geq 1$ and for $1 \leq t \leq \tau$.
Then
\[
\E_t[\Gamma_t] \geq \bigg( \frac{p(\lambda - 1) \lambda^p}{M} - O\Big( \frac{p^2 d}{kM} \Big) \bigg) \cdot \tr(\mX^{-p}).
\]
In particular, when $\lambda > 1 + \eps$ and $k \gtrsim \frac{p d}{\eps}$, the expected progress is positive.
\end{lemma}

\begin{proof}
Combining the expected loss \Cref{l:loss-p} and the expected gain~\Cref{l:gain-p}, it follows that 
\begin{align*}
\E_t[\Gamma_t] & \geq \frac{p}{M} \Big( \inner{\mX}{\mY_t^{-p-1}} - \tr(\mY_t^{-p}) \Big) - O\Big( \frac{p^2 d}{kM} \Big) \cdot \tr(\mX^{-p}).
\end{align*}
Applying \Cref{l:holder}, we derive that
\begin{align*}
\E_t[\Gamma_t] & \geq \frac{p}{M} \bigg( \Big(\frac{\tr(\mY_t^{-p})}{\tr(\mX^{-p})}\Big)^{1/p} \cdot \tr(\mY_t^{-p}) - \tr(\mY_t^{-p}) \bigg) - O\Big( \frac{p^2 d}{kM} \Big) \cdot \tr(\mX^{-p}).
\end{align*}
By the assumption that $(\tr(\mY_t^{-p}))^{1/p} \geq \lambda (\tr(\mX^{-p}))^{1/p}$, or equivalently $\tr(\mY_t^{-p}) \geq \lambda^p \tr(\mX^{-p})$, we arrive at the final bound that
\begin{align*}
\E_t[\Gamma_t] & \geq \frac{p(\lambda - 1)}{M} \cdot \tr(\mY_t^{-p}) - O\Big( \frac{p^2 d}{kM} \Big) \cdot  \tr(\mX^{-p}) \geq \bigg( \frac{p(\lambda - 1) \lambda^p}{M} - O\Big( \frac{p^2 d}{kM} \Big) \bigg) \cdot \tr(\mX^{-p}). \qedhere
\end{align*}
\end{proof}

\subsection{Martingale Concentration Argument} \label{ss:concentration}

In this subsection, we prove that the total progress is concentrated around the expectation.
The proof uses the minimum eigenvalue assumption and the optimality characterization in Lemma~\ref{l:optimality-p} to bound the variance of the random process.
The proof idea is similar to the one in~\cite{LZ21}, but we need some additional efforts to take care of the higher order terms that are introduced by higher $p$-norm.

\begin{lemma} \label{l:progress-whp-p}
For any $\eta > 0$, it holds that
\begin{multline*}
\Pr \left[ \sum_{t=1}^\tau \Gamma_t \leq \sum_{t=1}^\tau \E_t[\Gamma_t] - \eta ~\bigcap~ \min_{1 \leq t \leq \tau} \lambda_{\min}(\mZ_t) \geq 1-5\gamma \right] \\
\leq \exp\left( -\Omega\left(\frac{\eta^2 k M}{\tau p^3 \sqrt{d} (\tr(\mX^{-p}))^2 + \eta p M \tr(\mX^{-p})} \right)\right).
\end{multline*}
\end{lemma}

\begin{proof}
We define two sequences of random variables $\{X_t\}_t$ and $\{Y_t\}_t$, where $X_t := \E_t[\Gamma_t] - \Gamma_t$ and $Y_t := \sum_{l=1}^t X_l$.
It is easy to check that $\{Y_t\}_t$ is a martingale with respect to $\{S_t\}_t$.
We will use Freedman's inequality (\Cref{t:Freedman}) to bound the probability $\Pr[Y_\tau \geq \eta \cap \min_{1 \leq t \leq \tau} \lambda_{\min}(\mZ_t) \geq 1-5\gamma]$.

In the following, we first show that if the event $\min_{1 \leq t \leq \tau} \lambda_{\min}(\mZ_t) \geq 1-5\gamma$ happens, then we can upper bound $X_t$ and $\E_t[X_t^2]$ so that we can apply Freedman's inequality.
To upper bound $X_t$, we first prove an upper bound on $g_t$ and $l_t$.

Note that, if the event $\lambda_{\min}(\mZ_t) \geq 1-5\gamma$ happens, then $\mY_t \succeq (1-5\gamma) \mX$, which implies
\[
\vu_{i_t}^\top \mY_t^{-p-1}\vu_{i_t} \leq (1-5\gamma)^{-p-1} \vu_{i_t}^\top \mX^{-p-1} \vu_{i_t} \quad \text{and} \quad \vu_{j_t}^\top \mY_t^{-p-1} \vu_{j_t} \leq (1-5\gamma)^{-p-1} \vu_{j_t}^\top \mX^{-p-1} \vu_{j_t}.
\]
By \Cref{obs:prob}, for all the exchange pairs $i_t, j_t$, it holds that $\vx(i_t), \vx(j_t) \in (0,1)$. Thus, we can apply \Cref{l:optimality-p} to show that $\inner{\mX^{-p-1}}{\vu_{i_t} \vu_{i_t}^\top} \leq \frac1k \cdot \tr(\mX^{-p})$ and $\inner{\mX^{-p-1}}{\vu_{j_t} \vu_{j_t}^\top} \leq \frac1k \cdot \tr(\mX^{-p})$.
Therefore, since $\gamma = 1/6p$, 
\begin{align} \label{eq:opt-cond}
\vu_{i_t}^\top \mY_t^{-p-1} \vu_{i_t} \leq O\Big( \frac1k \Big) \cdot \tr(\mX^{-p}) \quad \text{and} \quad \vu_{j_t}^\top \mY_t^{-p-1} \vu_{j_t} \leq O\Big( \frac1k \Big) \cdot \tr(\mX^{-p}).
\end{align}

We first give a deterministic bound on $g_t$.
Let $x = \vu_{j_t}^\top \mY_t^{-p-1} \vu_{j_t}$, according to \eqref{eq:gt1}, the gain term $g_t$ can be written as
\[
g_t = \vu_{j_t}^\top \mY_t^{-p-1} \vu_{j_t} \cdot \frac1x \cdot \bigg( 1 - \Big( 1 - \frac{x}{1+x} \Big)^p \bigg).
\]
Since $(1-y)^p \geq 1-py$ for $y \in [0,1]$ and $p \geq 1$, we can bound $g_t$ by
\[
0 \leq g_t \leq \vu_{j_t}^\top \mY_t^{-p-1} \vu_{j_t} \cdot \frac1x \cdot \frac{px}{1+x} \leq p \cdot \vu_{j_t}^\top \mY_t^{-p-1} \vu_{j_t} \leq O\Big( \frac{p}{k} \Big) \cdot \tr(\mX^{-p}),
\]
where the last inequality follows from~\eqref{eq:opt-cond}.

Then, we give an deterministic bound on $l_t$.
By the definition of $S'_{t-1}$ and \eqref{eq:Y-Z}, it holds that $0 < \vu_{i_t}^\top \mY_t^{-1} \vu_{i_t} = \inner{\vv_{i_t} \vv_{i_t}^\top}{\mZ_t^{-1}} \leq \kappa$. 
Thus, we can bound $l_t$ by
\begin{align*}
0 \leq l_t & = \sum_{i=1}^p \binom{p}{i} \frac{(\vu_{i_t}^\top \mY_t^{-1} \vu_{i_t})^{i-1} \cdot \vu_{i_t}^\top \mY_t^{-p-1} \vu_{i_t}}{(1-\vu_{i_t}^\top \mY_t^{-1} \vu_{i_t})^i} \\
& \leq u_{i_t}^\top \mY_t^{-p-1} \vu_{i_t} \cdot \sum_{i=1}^p \binom{p}{i} \frac{\kappa^{i-1}}{(1-\kappa)^i} \\
& = \vu_{i_t}^\top \mY_t^{-p-1} \vu_{i_t} \cdot \frac{1}{\kappa} \cdot \bigg( \frac{1}{(1-\kappa)^p} - 1\bigg).
\end{align*}
For $\kappa = \frac{1}{2p}$, we can control $l_t$ such that
\[
l_t \leq O(p) \cdot \vu_{i_t}^\top \mY_t^{-p-1} \vu_{i_t} \leq O\Big( \frac{p}{k} \Big) \cdot \tr(\mX^{-p}), 
\]
where the last inequality follows from \eqref{eq:opt-cond}.

With the above bounds on $g_t$ and $l_t$, we can control the size of the martingale increment by
\[
|X_t| = |\E_t[\Gamma_t] - \Gamma_t|
\leq g_t + l_t 
\leq O\Big(\frac{p}{k} \Big) \cdot \tr(\mX^{-p}).
\]

Next, we upper bound $\E_t[X_t^2]$ by
\begin{equation*} 
\begin{aligned} 
\E_t[X^2_t] 
& \leq |X_t| \cdot \E_t[|X_t|] \leq O\Big(\frac{p}{k} \Big) \cdot \tr(\mX^{-p}) \cdot  \Big(\E_t[g_t ] + \E_t[l_t]\Big).
\end{aligned}
\end{equation*}
Using \Cref{l:loss-p}, we bound the expected loss term by
\begin{align*}
\E_t[l_t] 
& \leq \frac{p}{M} \Big( \tr(\mY_t^{-p}) - \inner{\mX_{S_{t-1}}}{\mY_t^{-p-1}} \Big) + O\Big( \frac{p^2 d}{kM} \Big) \cdot \tr(\mX^{-p}) \\
& \leq O\Big( \frac{p }{M}\Big) \cdot \tr(\mX^{-p}) + O\Big( \frac{p^2 d}{kM} \Big) \cdot \tr(\mX^{-p}) \\ 
& \leq O\Big( \frac{p}{M} \Big) \cdot \tr(\mX^{-p}),
\end{align*}
where the inequality on the second line follows by the assumption that $\lambda_{\min}(\mZ_t) \geq 1-5\gamma$ happens (and \eqref{eq:Y-X}),
the inequality on the last line follows by the assumption on $k$. 

Then, with the expression of $g_t$ in~\eqref{eq:gt1}, we write the expected gain as
\begin{align*}
\E_t[g_t] 
= \sum_{j \in [m] \backslash S_{t-1}} \frac{\vx(j)}{M} 
\cdot \Big( 1+ \alpha \inner{\vv_j \vv_j^\top}{\mA_t^{\frac12}} \Big) 
\cdot \frac{\vu_{j}^\top \mY_t^{-p-1} \vu_{j}}{\vu_{j}^\top \mY_t^{-1} \vu_{j}} \bigg( 1 - \bigg( 1 - \frac{\vu_{j}^\top \mY_t^{-1} \vu_{j}}{1 + \vu_{j}^\top \mY_t^{-1} \vu_{j}}\bigg)^p \bigg).
\end{align*}
Using the fact that $(1-y)^p \geq 1- py$ for $y \in [0,1]$ and $p\geq 1$, the expected gain can be bounded by
\begin{align*}
    \E_t[g_t] \leq \sum_{j \in [m] \backslash S_{t-1}} \frac{p}{M} 
\cdot \vx(j) \cdot \vu_j^\top \mY_t^{-p-1} \vu_j \cdot \frac{1+ \alpha \inner{\vv_j \vv_j^\top}{\mA_t^{\frac12}} }{1+\vu_j^\top \mY_t^{-1} \vu_j}.
\end{align*}
By \eqref{eq:Y-Z}, $\vu_j^\top \mY_t^{-1} \vu_j = \inner{\vv_j \vv_j^\top}{\mZ^{-1}}$. Then, by the second inequality in \Cref{l:E-A-gamma}, it holds that
\[
\E_t[g_t] \leq \frac{p}{M} \cdot \alpha \lambda_{\min}(\mZ_t) \cdot \sum_{j=1}^m \vx(j) \cdot \vu_j^\top \mY_t^{-p-1} \vu_j \leq \frac{p}{M} \cdot \alpha \cdot \inner{\mX}{\mY_t^{-p-1}},
\]
where the last inequality holds as $\lambda_{\min}(\mZ_t) < 1$ before the termination of the algorithm. By the assumption that $\lambda_{\min}(\mZ_t) \geq 1-5\gamma$ happens (and \eqref{eq:Y-X}) and the choice of $\alpha = \sqrt{d}/\gamma$ and $\gamma = 1/6p$, we obtain the bound that
\[
\E_t[g_t] \leq \frac{p \sqrt{d}}{\gamma (1-5\gamma)^{p+1} M} \cdot \tr(\mX^{-p}) \leq O\Big( \frac{p^2 \sqrt{d}}{M} \Big) \cdot \tr(\mX^{-p}).
\]

Therefore,
\[
\E_t[X^2_t] \leq O\Big( \frac{p}{k} \Big) \cdot O\Big( \frac{p^2 \sqrt{d}}{M} \Big) \cdot \big(\tr(\mX^{-p})\big)^2 = O\Big( \frac{p^3 \sqrt{d}}{k M} \Big) \cdot \big(\tr(\mX^{-p})\big)^2,
\]
which implies
\[
W_t := \sum_{l=1}^t \E_l[X^2_l] \leq  O\Big(  \frac{\tau p^3 \sqrt{d}}{k M}  \Big) \cdot \big(\tr(\mX^{-p})\big)^2, \quad \forall t \in [\tau].
\]
Finally, we can apply Freedman's martingale inequality Theorem~\ref{t:Freedman} with 
\[
R= O\Big(\frac{p}{k}\Big) \cdot \tr(\mX^{-p}) \qquad \text{and} \qquad \sigma^2 = O\Big(  \frac{p^3 \sqrt{d} \tau}{k M} \Big) \cdot \big(\tr(\mX^{-p})\big)^2
\]
to conclude that
\begin{align*}
\Pr\left[ Y_\tau \geq \eta ~\bigcap~ \min_{1 \leq t \leq \tau} \lambda_{\min}(\mZ_t) \geq 1-5\gamma \right] & \leq \Pr[\exists t \in [\tau]: Y_t \geq \eta \cap W_t \leq \sigma^2]\\
& \leq \exp\left(-\frac{\eta^2/2}{\sigma^2 + R\eta/3}\right) \\
& = \exp\left( -\Omega\left(\frac{\eta^2 k M}{\tau p^3 \sqrt{d} (\tr(\mX^{-p}))^2 + \eta p M \tr(\mX^{-p})} \right)\right).
\end{align*}
The lemma follows by noting that $\sum_{t=1}^\tau \Gamma_t \leq \sum_{t=1}^\tau \E_t[\Gamma_t] - \eta$ is equivalent to $Y_\tau \geq \eta$. 
\end{proof}

\subsection{Proof of \Cref{t:termination-p}} \label{ss:termination-p}

We are ready to prove Theorem~\ref{t:termination-p}.
Let $\tau = \frac{2M}{\eps p}$. 
We want to upper bound the probability that the following three events happen simultaneously:
\begin{itemize}
    \item $E_1$: The randomized exchange algorithm entered the second phase, i.e., $\lambda_{\min}(\mZ_1) \geq 1-2\gamma$ using the notation in this subsection.
    \item $E_2$: $\min_{1 \leq t \leq \tau} \lambda_{\min}(\mZ_t) \geq 1-5\gamma$.
    \item $E_3$: The second phase of the algorithm has not terminated by time $\tau$.
\end{itemize}

Suppose the event $E_3$ happens. 
Then $\lambda = \min_{1 \leq t \leq \tau+1} (\tr(\mY_t^{-p}))^{1/p}/(\tr(\mX^{-p}))^{1/p} > (1+\eps)$.
If the event $E_2$ also happens, then \Cref{l:p-exp-progress} implies that 
\begin{align} \label{eq:exp-progress-p}
\sum_{t=1}^{\tau} \E_t \left[ \Gamma_t \right] \geq \tau \cdot \bigg( \frac{p(\lambda - 1) \lambda^p}{M} - O\Big( \frac{p^2 d}{kM} \Big) \bigg) \cdot \tr(\mX^{-p}) \geq \frac{2M}{\eps p} \cdot \bigg( \frac{\eps p}{M} - \frac{\eps p}{2M} \bigg) \cdot \tr(\mX^{-p})
\geq \tr\left(\mX^{-p}\right),
\end{align}
where the second inequality holds for $k \gtrsim p d / \eps$ with large enough constant.

On the other hand, the initial solution of the second phase satisfies $\mZ_1 \succcurlyeq (1-2\gamma) \mI$ (follows by $E_1$), which implies that $\mY_1 \succcurlyeq (1-2\gamma) \mX$. 
Thus, $\tr(\mY_1^{-p}) \leq (1-2\gamma)^{-p} \tr(\mX^{-p}) \leq (1-1/3p)^{-p} \tr(\mX^{-p}) \leq 3 \tr(\mX^{-p}) /2$, for $p \geq 1$.
When the event $E_2$ happens, we know from Lemma~\ref{l:E-A-gamma} that 
$\inner{\vv_{i_t} \vv_{i_t}^\top}{\mZ_t^{-1}} \leq \alpha \cdot \inner{\vv_{i_t} \vv_{i_t}^\top}{\mA_t^{\frac12}} < \frac12$,
and so we can apply~\eqref{eq:trace-bound-p} to deduce that
\begin{align*}
\tr(\mY_{\tau+1}^{-p}) \leq \tr(\mY_{1}^{-p}) - \sum_{t=1}^\tau \Gamma_t \leq \frac32 \cdot \tr\big(\mX^{-p}\big) - \sum_{t=1}^\tau \Gamma_t.
\end{align*}
Since the algorithm has not terminated by time $\tau$,
\begin{align} \label{eq:progress-up-p}
\tr(\mX^{-p}) \leq \tr(\mY_{\tau+1}^{-p}) \quad \implies \quad \sum_{t=1}^\tau \Gamma_t \leq \frac12 \cdot \tr\big(\mX^{-p}\big). 
\end{align}
Combining~\eqref{eq:exp-progress-p} and~\eqref{eq:progress-up-p}, $E_1 \cap E_2 \cap E_3$ implies a large deviation of the progress from the expectation such that
\[
\sum_{t=1}^\tau \Gamma_t - \sum_{t=1}^\tau \E_t[\Gamma_t] < -\frac12 \cdot \tr(\mX^{-p}).
\]
Thus, we can apply Lemma~\ref{l:progress-whp-p} with $\eta=\frac12 \cdot \tr(\mX^{-p})$ and $\tau=\frac{2M}{\eps p}$ to conclude that
\begin{equation*}
\begin{aligned}
\Pr\left[ E_1 \cap E_2 \cap E_3  \right] \leq & \Pr\left[ \sum_{t=1}^{\tau} \Gamma_t <  \sum_{t=1}^{\tau} \E_t \left[ \Gamma_t \right] - \frac12 \cdot \tr(\mX^{-p}) ~\bigcap~ E_2 \right]
\\
\leq & \exp\bigg( -\Omega\bigg(\frac{(\tr(\mX^{-p}))^2 \cdot k M }{\left(\frac{2M}{\eps p}\right) p^3 \sqrt{d} \cdot (\tr(\mX^{-p}))^2 + p M \cdot (\tr(\mX^{-p}))^2} \bigg)\bigg)
\\
\leq & \exp\left(-\Omega\Big( \frac{\eps k}{p^2 \sqrt{d}} \Big)\right) \leq \exp\left(-\Omega\Big( \frac{\sqrt{d}}{p} \Big)\right) = \exp\left(-\Omega\Big( \gamma \sqrt{d} \Big)\right),
\end{aligned}
\end{equation*}
where the last inequality holds by the assumption $k \gtrsim p d/\eps$, and the last equality follows as $\gamma = 1/6p$.

\section*{Acknowledgment}
We thank Mohit Singh for bringing the $\Phi_p$ objective function to our attention.
We also thank anonymous reviewers of an earlier version of this manuscript for helpful suggestions.

\bibliographystyle{plain}
\bibliography{references}

\appendix

\section{Background}

We will use the following two results to analyze the $\Phi_p$ objective function.

\begin{lemma}[Sherman-Morrison Formula~\cite{SM50}] \label{l:Sherman-Morrison}
Suppose $\mA \in \R^{d \times d}$ is an invertible matrix, and $\vu, \vv \in \R^d$. Then, $\mA + \vu \vv^\top$ is invertible if and only if $1 + \vv^\top \mA^{-1} \vu \neq 0$, and under this case
\[
\left( \mA + \vu \vv^\top \right)^{-1} = \mA^{-1} - \frac{\mA^{-1} \vu \vv^\top \mA^{-1}}{1 + \vv^\top \mA^{-1} \vu}.
\]
\end{lemma}

\begin{lemma}[Lieb-Thirring Inequality~\cite{LT76}] \label{l:Lieb-Thirring}
For any $\mA, \mB \succcurlyeq 0$ and $p \geq 1$, it holds that
\[
\tr((\mB \mA \mB)^p) \leq \tr(\mB^p \mA^p \mB^p).
\]
\end{lemma}

We will use concentration inequality for martingales to analyze our randomized algorithm.
Recall that a sequence of random variables $Y_1, \ldots, Y_\tau$ is a martingale with respect to a sequence of random variables $Z_1, \ldots, Z_\tau$ if for all $t > 0$, it holds that 
\begin{enumerate}
\item $Y_t$ is a function of $Z_1, \ldots, Z_{t-1}$;
\item $\E[|Y_t|] < \infty$; 
\item $\E[Y_{t+1} | Z_1, \ldots, Z_t] = Y_t$. 
\end{enumerate}
We will use the following theorem by Freedman to bound the probability that $Y_\tau$ is large.

\begin{theorem}[\cite{Fre75,Tro11}] \label{t:Freedman}
Let $\{Y_t\}_t$ be a real-valued martingale with respect to $\{Z_t\}_t$,
and $\{X_t=Y_t - Y_{t-1}\}_t$ be the difference sequence. 
Assume that $X_t \leq R$ deterministically for $1 \leq t \leq \tau$. 
Let $W_t := \sum_{j=1}^t \E[X_j^2 | Z_1, ..., Z_{j-1}]$ for $1 \leq t \leq \tau$.
Then, for all $\delta \geq 0$ and $\sigma^2 > 0$,
\[
\Pr\left( \exists t \in [\tau]: Y_t \geq \delta~{\rm and}~W_t \leq \sigma^2 \right) \leq \exp\left( \frac{-\delta^2/2}{\sigma^2 + R\delta/3} \right).
\]
\end{theorem}

\section{Omitted Proofs in \Cref{s:framework}} \label{a:cost}

\begin{proof}[Sketched Proof of \Cref{t:cost-kappa}]
As in the proof of Theorem 3.12 in~\cite{LZ20}, there are two main parts in proving \Cref{t:cost-kappa}: (1) bound the expected change of solution size; (2) bound the conditional variance of the stochastic process.
The main difference in proving \Cref{t:cost-kappa} is within the first part, i.e., the expected change of solution size will become worse because of the new $\kappa$ parameter.

Let ${\rm 1}_{i_t}$ and ${\rm 1}_{j_t}$ be the indicator variables of whether $i_t$ and $j_t$ are empty or not.
The key of the updated proof is to show the following inequalities.
\begin{align} \label{eq:expected-change}
\frac{1}{M} \Big(\sum_{i=1}^n \vx(i) - |S_{t-1}| \Big) \leqslant \E_t[{\rm 1}_{j_t} - {\rm 1}_{i_t}  ]
\leqslant \frac{1}{M} \bigg(\sum_{i=1}^n \vx(i) - |S_{t-1}| + \Big( \frac{11d}{\gamma} + \frac{d}{\kappa}\Big) \bigg).
\end{align}
We will sketch how we obtain this bound and highlight the loss caused by $\kappa < 1$.

We first bound the conditional expectation of ${\rm 1}_{j_t}$.
With the same argument as in~\cite{LZ20}, we can show that
\begin{equation} \label{eq:expected-add}
\frac{1}{M} \bigg( \sum_{i=1}^n \vx(i) - \sum_{i \in S_{t-1}} \vx(i) \bigg)
\leqslant \E_t[{\rm 1}_{j_t} ] 
\leqslant \frac{1}{M} \bigg( \sum_{i=1}^n \vx(i)  - \sum_{i \in S_{t-1}} \vx(i) + \frac{d\norm{\vc}_{\infty}}{\gamma}\bigg).
\end{equation}
Next we bound the expectation of ${\rm 1}_{i_t}$. This is the main difference caused by $\kappa$.
By the probability distribution of $i_t$,
\begin{align*}
\E_t[{\rm 1}_{i_t} ] 
= &~ \frac{1}{M} \sum_{i \in S'_{t-1}} (1-\vx(i)) \big(1-\alpha \inner{\vv_i \vv_i^{\top}}{\mA_t^{\frac12}} \big)
\\ 
= &~ \frac{1}{M} \Big(\sum_{i \in S'_{t-1}}  (1-\vx(i)) - \alpha \sum_{i \in S'_{t-1}}  (1-\vx(i)) \inner{\vv_i \vv_i^{\top}}{\mA_t^{\frac12}} \Big)
\\
= &~ \frac{1}{M} \bigg( |S_{t-1}| - \bigg( \sum_{i \in S_{t-1}} \vx(i) \bigg) - \sum_{i \in S_{t-1} \backslash S'_{t-1}} (1-\vx(i))
- \alpha \sum_{i \in S'_{t-1}} (1-\vx(i)) \inner{\vv_i \vv_i^{\top}}{\mA_t^{\frac12}} \bigg).
\end{align*}
We would like to bound the last two terms of the right hand side.
Recall that $S'_{t-1} := \{ i \in S_{t-1} \mid \alpha\inner{\vv_i \vv_i^{\top}}{\mA_t^{\frac12}} \leq \frac12~\text{and}~ \inner{\vv_i \vv_i^\top}{\mZ_t^{-1}} \leq \kappa \}$. Therefore, for any $i \in S_{t-1} \backslash S'_{t-1}$, it holds that $1 < 2\alpha \inner{\vv_i \vv_i^\top}{\mA_t^{1/2}} + \kappa^{-1} \inner{\vv_i \vv_i^\top}{\mZ_t^{-1}}$.
This implies that
\begin{align*}
|S_{t-1} \setminus S'_{t-1}|
& \leq \sum_{i \in S_{t-1} \setminus S'_{t-1}} 2 \alpha \inner{\vv_i \vv_i^{\top}}{\mA_t^{\frac12}} + \sum_{i \in S_{t-1} \backslash S'_{t-1}} \kappa^{-1} \inner{\vv_i \vv_i^\top}{\mZ_t^{-1}} \\
& \leq 2 \alpha \sum_{i \in S_{t-1}} \inner{\vv_i \vv_i^{\top}}{\mA_t^{\frac12}} + \kappa^{-1} \sum_{i \in S_{t-1}} \inner{\vv_i \vv_i^\top}{\mZ_t^{-1}} \\
& \leqslant 2 \big(d + \alpha \sqrt{d} \cdot \lambda_{\min}(\mZ_t)\big) + \frac{d}{\kappa} \leqslant \frac{8d}{\gamma} + \frac{d}{\kappa},
\end{align*}
where the second last inequality uses Claim 2.14 in~\cite{AZLSW21} and the last inequality is by $\alpha = \frac{\sqrt{d}}{\gamma}$ and the assumption that $\lambda_{\min}(\mZ_t) < 1$.
Since $\vx \in [0,1]^m$, it follows that the second last term is 
\[
0 \leqslant \sum_{i \in S_{t-1} \backslash S'_{t-1}} (1-\vx(i)) \leqslant |S_{t-1} \backslash S'_{t-1}| \leqslant \Big( \frac{8d}{\gamma} + \frac{d}{\kappa} \Big) 
\]
Similarly, for the last term,
\begin{align*}
0 \leqslant \alpha \sum_{i \in S'_{t-1}} (1-\vx(i)) \inner{\vv_i \vv_i^{\top}}{\mA_t^{\frac12}} 
& \leqslant \alpha \sum_{i \in S_{t-1}} \inner{\vv_i \vv_i^{\top}}{\mA_t^{\frac12}} \leqslant(d+\alpha\sqrt{d}\cdot \lambda_{\min}(\mZ_t)) 
\leqslant \frac{2d}{\gamma}.
\end{align*}
Plugging back these upper and lower bounds for the last two terms, for $\kappa < 1$, we obtain
\begin{equation} \label{eq:expected-remove}
\frac{1}{M} \Bigg( |S_{t-1}| - \bigg(\sum_{i \in S_{t-1}} \vx(i) \bigg) - \Big(\frac{10d}{\gamma} + \frac{d}{\kappa} \Big) \Bigg)
\leqslant \E_t[{\rm 1}_{i_t}] 
\leqslant \frac{1}{M} \bigg( |S_{t-1}| - \sum_{i \in S_{t-1}} \vx(i) \bigg).
\end{equation}
Combining \eqref{eq:expected-add} and \eqref{eq:expected-remove}, we can obtain \eqref{eq:expected-change}. The remaining concentration argument is the same as in the proof of Theorem 3.12 in~\cite{LZ21}.
\end{proof}

\section{Omitted Proofs in \Cref{s:objective}} \label{a:proofs-sec-term}

\begin{proof}[Proof of Lemma~\ref{l:E-A-gamma}]
Recall that $\mA_t = (\alpha \mZ_t - c_t \mI)^{-2}$ where $c_t$ is the unique value such that $\mA_t \succ 0$ and $\tr(\mA_t) = 1$. 
Since $\mZ_t \succcurlyeq \lambda_{\min}(\mZ_t) \cdot \mI$, it follows that
\[
1 = \tr(\mA_t) \leq \left(\alpha \lambda_{\min}(\mZ_t) - c_t\right)^{-2} \cdot \tr(\mI) 
\quad \Longrightarrow \quad 
\alpha \lambda_{\min}(\mZ_t) - c_t \leq \sqrt{d} 
\quad \Longrightarrow \quad 
c_t \geq 0,
\]
where the last implication holds as $\alpha = \sqrt{d}/\gamma$, $\lambda_{\min}(\mZ_t) \geq 1-5\gamma$ and the assumption $\gamma \leq \frac16$.
This implies that $\mA^{\frac12}_t = (\alpha \mZ_t - c_t \mI)^{-1} \succcurlyeq \alpha^{-1} \mZ^{-1}_t$, proving the first inequality.

For the second inequality, consider the eigen-decomposition of $\mZ_t \!=\! \sum_{j =1}^d \lambda_j \vw_j \vw_j^\top$, 
where $0 < \lambda_1 \leq \ldots \leq \lambda_d$ are the eigenvalues and $\{\vw_j\}$ are the corresponding orthonormal eigenvectors.
Then,
\begin{align*}
    \frac{\inner{\vv_i \vv_i^\top}{\mA_t^{\frac12}}}{\inner{\vv_i \vv_i^\top}{\mZ_t^{-1}}} & = \frac{\sum_{j=1}^d \frac{\inner{\vv_i}{\vw_j}^2}{\alpha \lambda_j - c_t}}{\sum_{j=1}^d \frac{\inner{\vv_i}{\vw_j}^2}{\lambda_j}} \leq \max_{j \in [d]} \frac{\lambda_j}{\alpha \lambda_j - c_t} \leq \frac{\lambda_1}{\alpha \lambda_1 - c_t} \leq \lambda_1,
\end{align*}
where the first inequality holds since $\alpha \lambda_j - c_t > 0$ as $\mA_t \succ 0$, 
the second inequality holds as $c_t \geq 0$ and the function $f(x) = \frac{x}{\alpha x - c_t}$ is decreasing for $x \geq \frac{l_t}{\alpha}$ when $c_t \geq 0$,
and the last inequality follows as $1=\tr(\mA_t) \geq (\alpha \lambda_1 - c_t)^{-2}$ which implies $\alpha \lambda_1 - c_t \geq 1$. 
\end{proof}

\begin{proof}[Proof of \Cref{l:optimality-p}]
We will use a dual characterization to prove the lemma.
We introduce a dual variable $\mY$ for the first equality constraint, 
a dual variable $\mu \geq 0$ for the budget constraint $\sum_{i=1}^m \vx(i) - k \leq 0$, 
a dual variable $\beta^-_i \geq 0$ for each non-negative constraint $-\vx(i) \leq 0$,
and a dual variable $\beta^+_i \geq 0$ for each capacity constraint $\vx(i) - 1 \leq 0$.

The Lagrange function $L(\vx, \mX, \mY, \mu, \beta^+, \beta^-)$ is defined as
{\small
\begin{align*}
& \Big(\tr\big(\mX^{-p}\big)\Big)^{1/p}
\!+\! \bigg\langle \mY, \mX - \sum_{i=1}^m \vx(i) \cdot \vu_i \vu_i^\top \bigg\rangle 
\!+\! \mu \bigg(\sum_{i=1}^m \vx(i) - k\bigg)
 \!-\! \sum_{i=1}^m \beta^-_i \vx(i) + \sum_{i=1}^m \beta^+_i (\vx(i)-1) 
\\
= & \Big(\tr\big(\mX^{-p}\big)\Big)^{1/p} \!+\! \inner{\mY}{\mX} \!-\! \mu k
\!-\! \sum_{i=1}^m \beta^+_i
\!-\! \sum_{i=1}^m \vx(i) \cdot \big(\inner{\mY}{\vu_i \vu_i^\top} - \mu + \beta^-_i - \beta^+_i \big).
\end{align*}}%
The Lagrangian dual program is 
\[
\max_{\substack{\mY \succ 0, \mu \geq 0, \\ \beta^+ \geq 0, \beta^- \geq 0}} \min_{\mX \succ 0, \vx} L(\vx,\mX,\mY,\mu,\beta^+,\beta^-).
\]
Note that we can add the constraint $\mX \succ 0$ to the inner minimization without loss of generality, as otherwise the Lagrange function is unbounded. Similarly, we can assume that $\mY \succ 0$ in the outer maximization, as otherwise the inner maximization problem is unbounded below.

To simplify the dual program, we will apply first order optimality condition to the inner minimization problem. We start with computing the gradient of $(\tr(\mX^{-p}))^{1/p}$.
\begin{align*}
\nabla_\mX (\tr(\mX^{-p}))^{1/p} = \frac1p (\tr(\mX^{-p}))^{\frac{1-p}{p}} \cdot \partial_\mX (X^p)^{-1} & = -\frac1p (\tr(\mX^{-p}))^{\frac{1-p}{p}} \cdot X^{-p} (\partial_\mX X^p) X^{-p} \\
& = -(\tr(\mX^{-p}))^{\frac{1-p}{p}} \cdot X^{-p-1}.
\end{align*}
Therefore, given any $\mY \succ 0, \mu, \beta^+, \beta^- \geq 0$, the minimizers $\mX \succ 0, \vx$ of the inner minimization problem satisfy the optimality conditions that
\begin{align} \label{eq:optimality}
\nabla_{\mX} L = -(\tr(\mX^{-p}))^{\frac{1-p}{p}} \cdot X^{-p-1} + \mY = 0 \quad \text{and} \quad \nabla_{\vx(i)} L = -\inner{\mY}{\vu_i \vu_i^\top} + \mu - \beta^-_i + \beta^+_i = 0.
\end{align}
On the other hand, it is easy to verify that $\vx = \delta \vec{1}$ is a strictly feasible solution of the primal program for a small enough $\delta$. 
So, Slater's condition implies that strong duality holds. 
Let $\vx,\mX$ be an optimal solution for the primal program,
and $\mY,\mu,\beta^+,\beta^-$ be an optimal solution for the dual program. Strong duality implies that
\[
\Big(\tr\big(\mX^{-p}\big)\Big)^{1/p} = \Big(\tr\big(\mX^{-p}\big)\Big)^{1/p} \!+\! \inner{\mY}{\mX} \!-\! \mu k
\!-\! \sum_{i=1}^m \beta^+_i
\!-\! \sum_{i=1}^m \vx(i) \cdot \big(\inner{\mY}{\vu_i \vu_i^\top} - \mu + \beta^-_i - \beta^+_i \big),
\]
where the LHS is the optimal value of the primal program and the RHS is equal to the optimal value of the dual program given the optimal primal and dual solutions.

Then, it follows from the Lagrangian optimality condition~\eqref{eq:optimality} that
\begin{align*}
0 = \tr(\mX^{-p})^{1/p} - \mu k - \sum_{i=1}^m \beta^+_i 
\quad & \Longrightarrow \quad 
\mu k + \sum_{i=1}^m \beta^+_i = \tr(\mX^{-p})^{1/p} \quad \Longrightarrow \quad 
\mu \leq \frac{1}{k} \cdot \tr(\mX^{-p})^{1/p},
\end{align*}
where the last implication follows from $\beta^+_i \geq 0$ for all $i \in [m]$.

Finally, by the complementary slackness conditions, we have $\beta^-_i \cdot \vx(i) = 0$ and $\beta^+_i \cdot (1-\vx(i)) = 0$ for all $i \in [n]$. 
Therefore, for each $i \in [n]$ with $0 < \vx(i) < 1$, we must have $\beta^+_i = \beta^-_i = 0$. By~\eqref{eq:optimality}, it follows that
\begin{align*}
\frac1k \cdot \tr(\mX^{-p})^{1/p} \geq \mu = \inner{\mY}{\vu_i \vu_i^\top} = (\tr(\mX^{-p}))^{\frac{1-p}{p}} \cdot \inner{X^{-p-1}}{\vu_i \vu_i^\top} ~\implies~ \inner{\mX^{-p-1}}{\vu_i \vu_i^\top} \leq \frac1k \cdot \tr(\mX^{-p}). \end{align*}
\end{proof}

\end{document}